%% file: AB13-INFOCOM-RR.tex
\documentclass[conference]{IEEEtran}
\makeatletter
\def\ps@headings{%
\def\@oddhead{\mbox{}\scriptsize\rightmark \hfil \thepage}%
\def\@evenhead{\scriptsize\thepage \hfil \leftmark\mbox{}}%
\def\@oddfoot{}%
\def\@evenfoot{}}
\makeatother
\usepackage[english]{babel}
\usepackage{amssymb}
\usepackage{amsmath,amsthm}
\usepackage{graphicx}
\usepackage{bbm}
\usepackage{url}
\usepackage[english, linesnumbered, ruled, vlined,figure]{algorithm2e}
\usepackage{pict2e}
\usepackage{wrapfig}
\usepackage{multirow}
\usepackage{subfigure}
\usepackage{placeins}
\usepackage{balance}

\newtheoremstyle{dotless-thm}
{3pt}
{3pt}
{\it}
{}
{\bfseries}
{}
{.5em}
{}

\theoremstyle{dotless-thm}
\newtheorem{theorem}{Theorem}
\newtheorem{property}[theorem]{Property}
\newtheorem{lemma}[theorem]{Lemma}
\newtheorem{definition}[theorem]{Definition}
\newtheorem{remark}[theorem]{Remark}

\begin{document}

\title{Sketch {\Huge ${\star}$}-metric: Comparing Data Streams via Sketching\\{\sc Research Report}}

\author{\IEEEauthorblockN{Emmanuelle Anceaume}
\IEEEauthorblockA{IRISA / CNRS, 
Rennes, France\\
\url{Emmanuelle.Anceaume@irisa.fr}}
\and
\IEEEauthorblockN{Yann Busnel}
\IEEEauthorblockA{LINA / Universit\'e de Nantes, 
Nantes, France\\
\url{Yann.Busnel@univ-nantes.fr}}
} 

\maketitle

\begin{abstract}
In this paper, we consider the problem of estimating the distance between any two large data streams in small-space constraint. This problem is of utmost importance in data intensive monitoring applications where input streams are generated rapidly. These streams need to be processed on the fly and accurately to quickly determine any deviance from nominal behavior. We present a new metric, the \emph{Sketch $\star$-metric}, which allows to define a distance between updatable summaries (or sketches) of large data streams. An important feature of the \emph{Sketch $\star$-metric} is that, given a measure on the entire initial data streams, the \emph{Sketch $\star$-metric} preserves the axioms of the latter measure on the sketch (such as the non-negativity, the identity, the symmetry, the triangle inequality but also  specific properties of the $f$-divergence). 
Extensive experiments conducted on both  synthetic  traces and real data allow us to validate the robustness and accuracy  of the \emph{Sketch $\star$-metric}. 
\end{abstract}


\IEEEpeerreviewmaketitle

\input{intro.tex}

\input{related-works.tex}

\input{model.tex}

\input{star-metric.tex}

\input{algorithm.tex}

\input{evaluation.tex}

\input{conclusion.tex}

\balance
\bibliographystyle{IEEEtran}
\bibliography{AB13-INFOCOM}  % sigproc.bib is the name of the Bibliography in this case

\begin{figure}[!h]
\centering
\subfigure[\emph{Sketch $\star$-metric} accuracy  as  a function of  parameter $t$. We have \newline
 $m= 200,000$; $n=4,000$; $k=200$ and  $r=3$]{
\includegraphics[width=.45\textwidth, height=0.21\textheight]{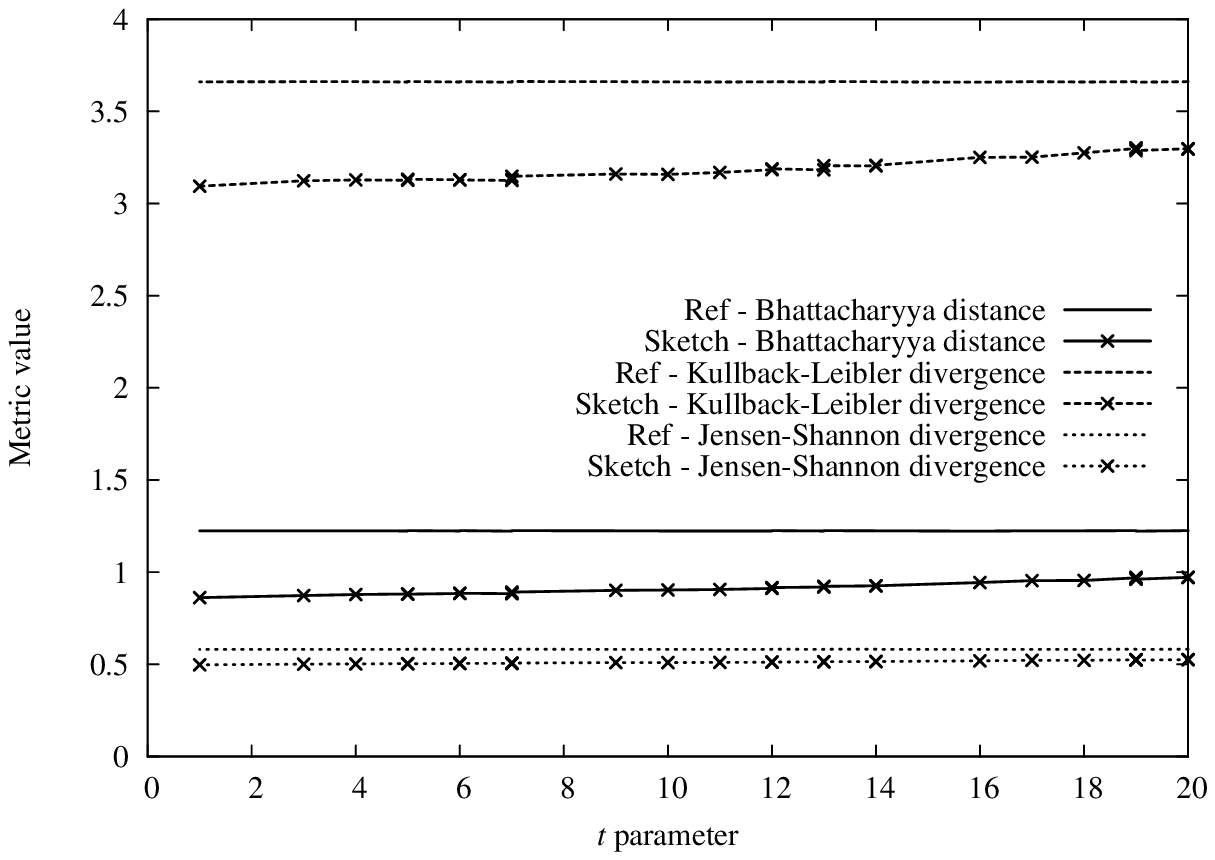}
\label{fig:test_pascal_delta}
}
\subfigure[{Sketch $\star$-metric} accuracy between data trace extracted from ClarkNetwork (August) and Saskatchewan University, as a function of $t$]{
\includegraphics[width=.45\textwidth, height=0.21\textheight]{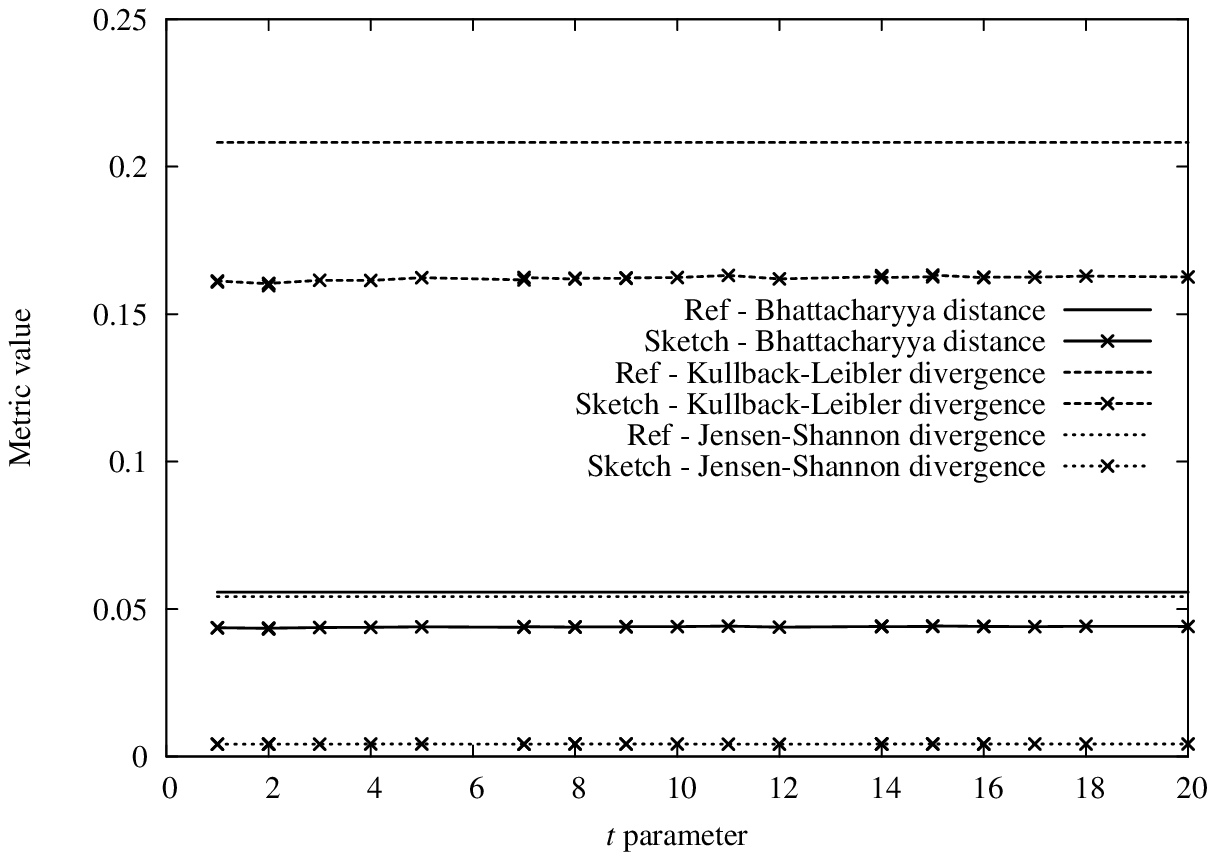}
\label{fig:test_real_delta}
}
\caption{\emph{Sketch $\star$-metric}  between the Uniform distribution and Pascal with parameter $p=\frac{n}{2r+n}$ (Figures~\ref{fig:test_pascal_epsilon} and~\ref{fig:test_pascal_delta}), and between data trace extracted from ClarkNetwork (August) and Saskatchewan University (Figures~\ref{fig:test_real_epsilon} and~\ref{fig:test_real_delta}).}  
\label{fig:tparameters}
\end{figure}

\begin{figure*}
\centering
\subfigure[Value of Bhattacharyya distance]{
\includegraphics[width=.45\textwidth, height=0.21\textheight]{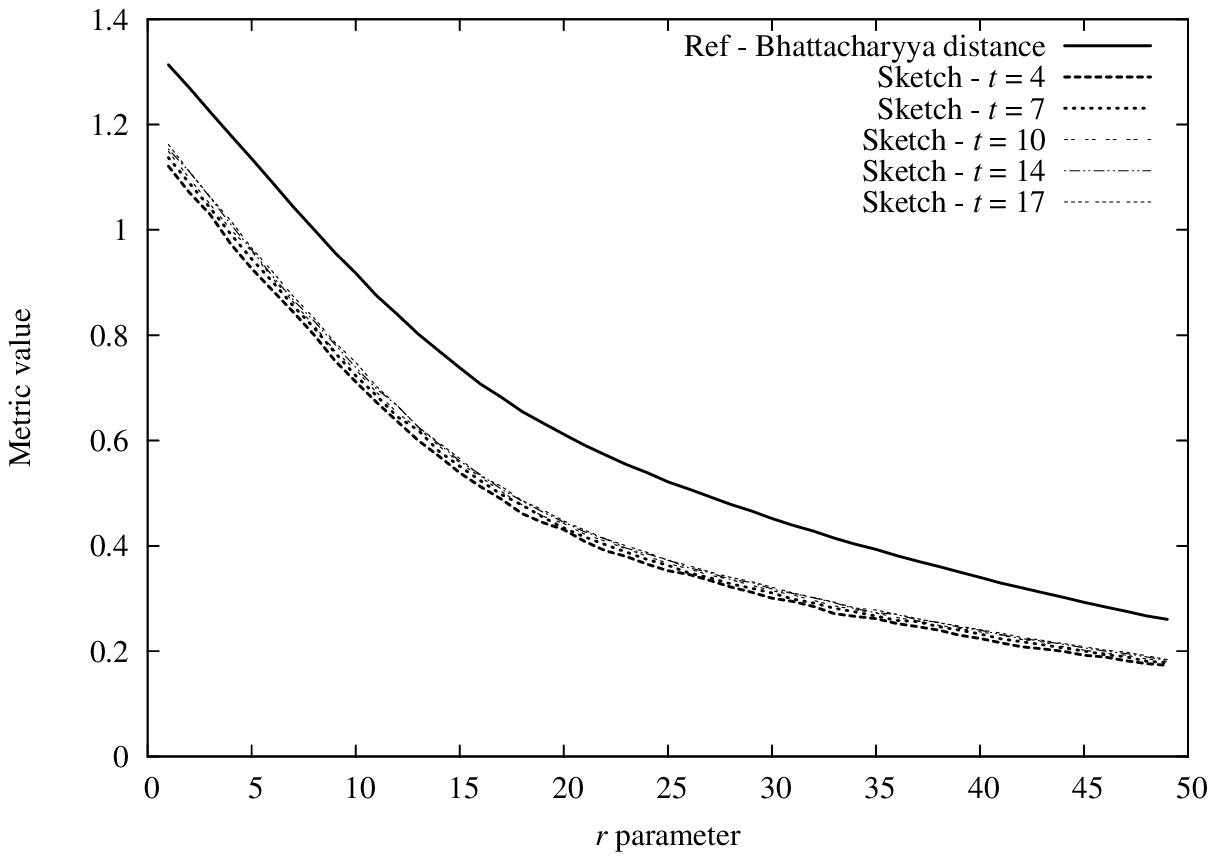}
\label{fig:test_pascal_r_delta_b}
}
\subfigure[Difference with Bhattacharyya distance]{
\includegraphics[width=.45\textwidth, height=0.21\textheight]{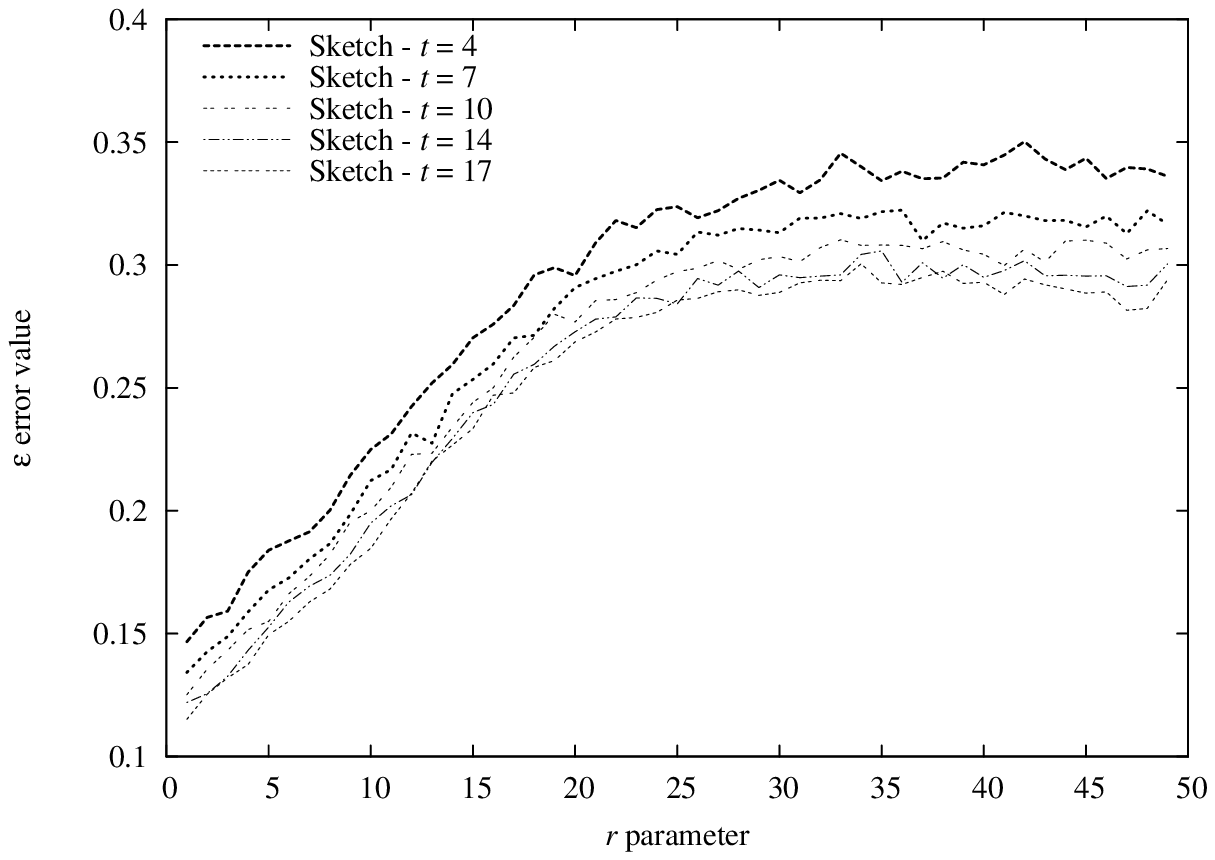}
\label{fig:test_pascal_r_delta_b_error}
}
\subfigure[Value of Kullback-Leibler divergence]{
\includegraphics[width=.45\textwidth, height=0.21\textheight]{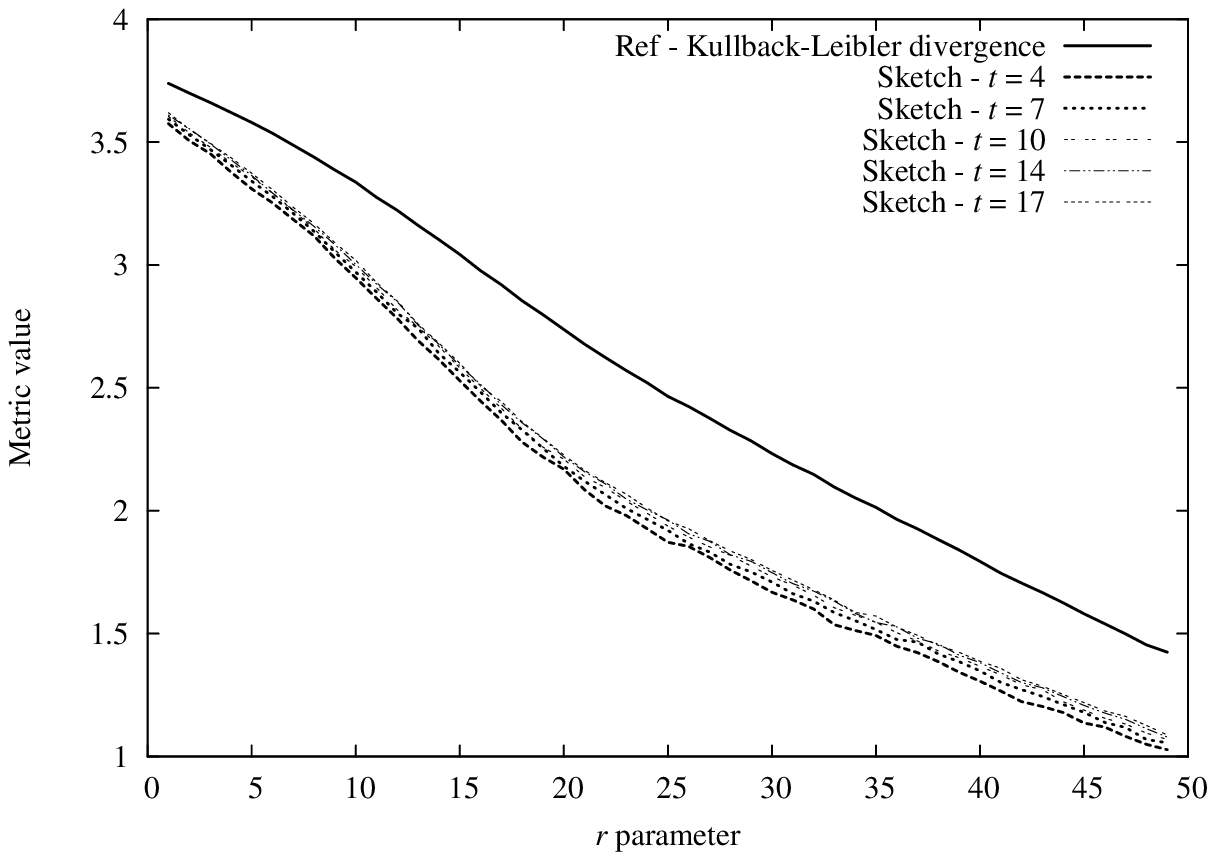}
\label{fig:test_pascal_r_delta_kl}
}
\subfigure[Difference with Kullback-Leibler divergence]{
\includegraphics[width=.45\textwidth, height=0.21\textheight]{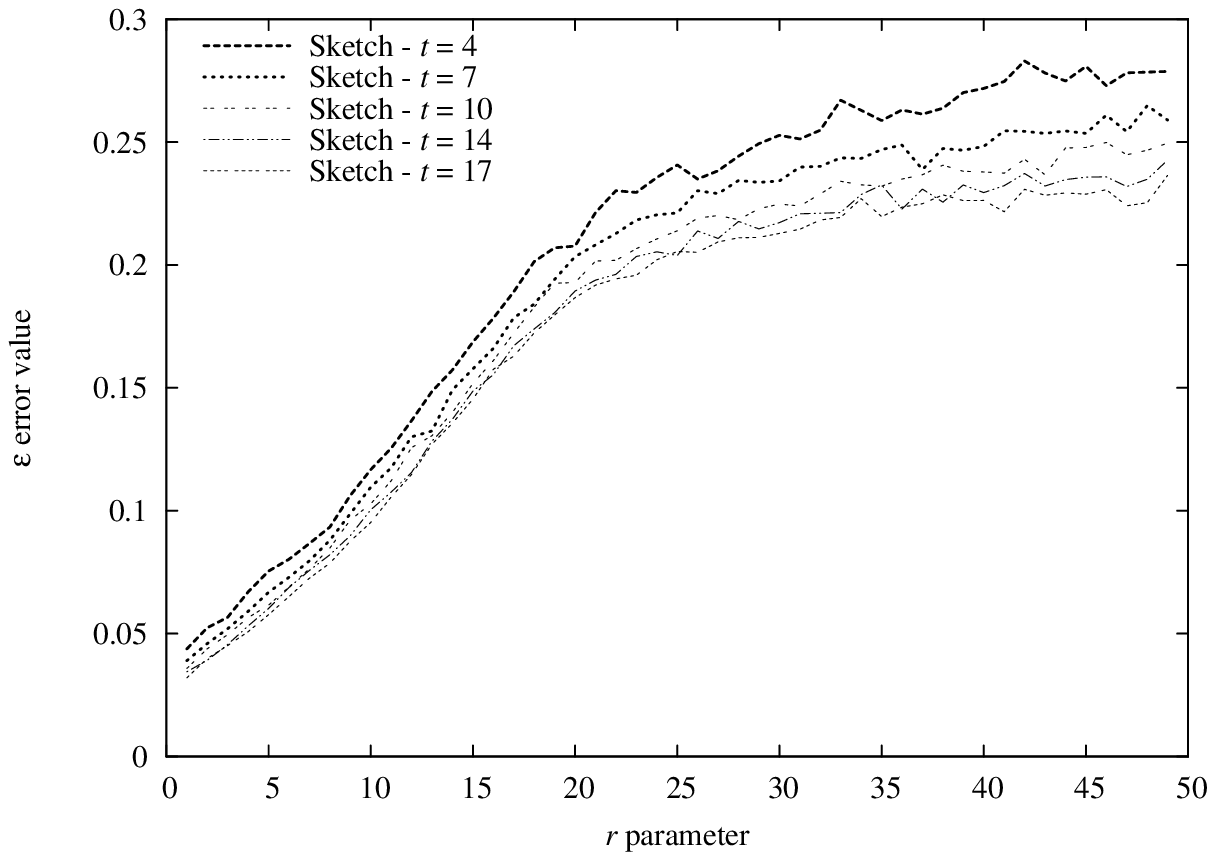}
\label{fig:test_pascal_r_delta_kl_error}
}
\subfigure[Value of Jensen-Shannon divergence]{
\includegraphics[width=.45\textwidth, height=0.21\textheight]{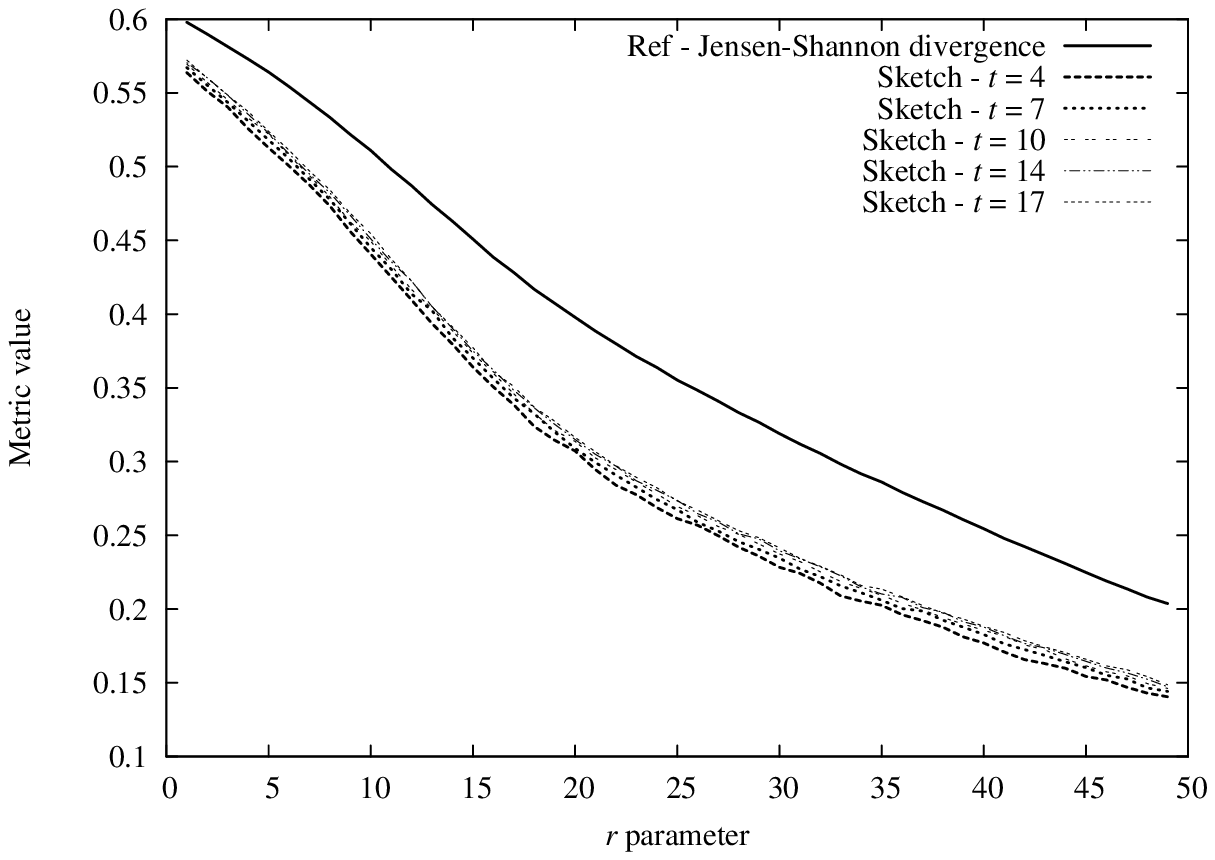}
\label{fig:test_pascal_r_delta_js}
}
\subfigure[Difference with Jensen-Shannon divergence]{
\includegraphics[width=.45\textwidth, height=0.21\textheight]{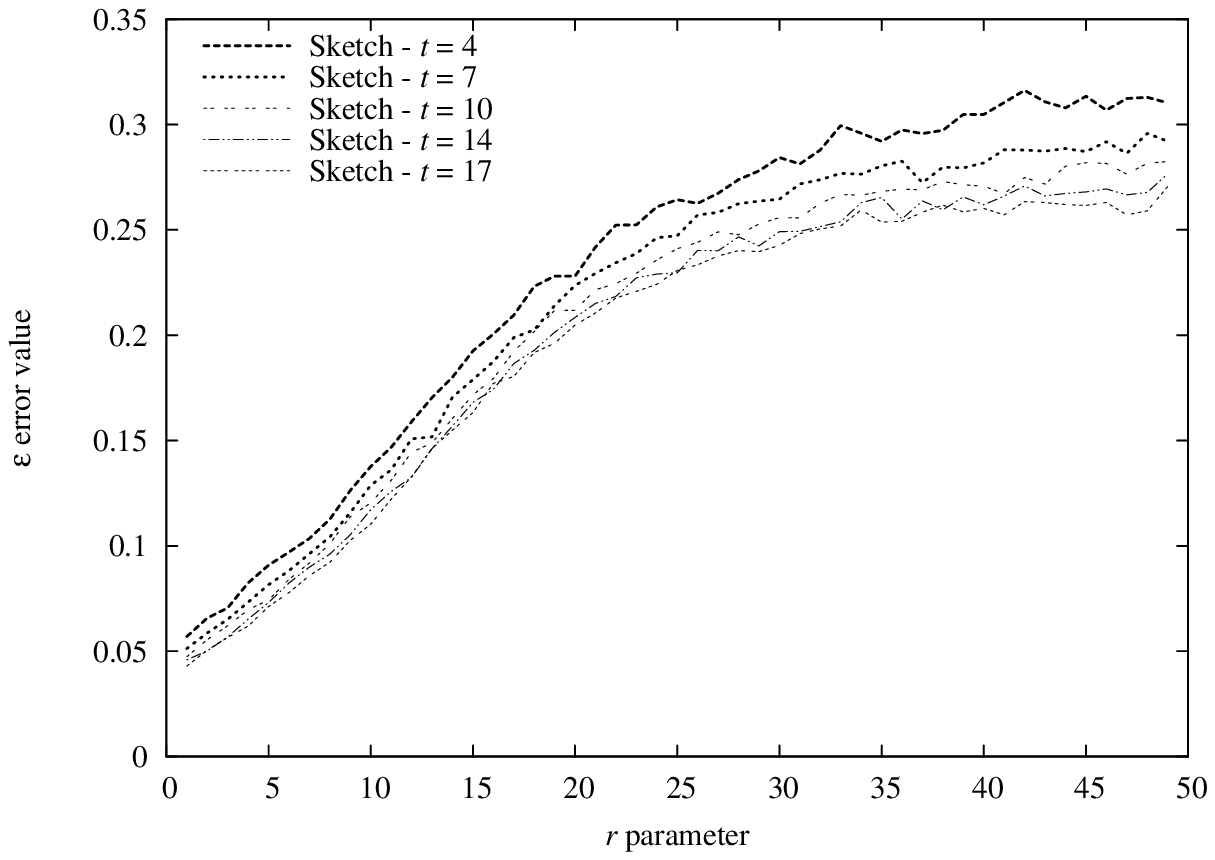}
\label{fig:test_pascal_r_delta_js_error}
}
\caption{\emph{Sketch $\star$-metric} estimation between Uniform distribution and Pascal with parameter $p=\frac{n}{2r+n}$, as a function of $k$, $t$ and $r$.}  
\label{fig:parameters}
\end{figure*}

%\end{document}  % This is where a 'short' article might terminate

% That's all folks!
\end{document}

%% file: intro.tex
%!TEX root =  AB13-INFOCOM.tex

\section{Introduction}\label{toc:intro}
%plan
%\begin{itemize}
%\item context of data intensive applications
%\item there is a need to efficiently and accurately analyze these data 
%\item use information theoretic distance
%\item but too demanding in terms of space and time
%\item data stream model seems to be a good candidate as a computational model 
%\item many statistics have already been done in this model (entropy, frequent items, frequency moments...
%\item what about distance measures ?? 
%\item we have proposed to estimate the relative entropy of an observed stream with a uniform one through sketching
% \item in this paper we generalize this work by proposing a metric computed on summaries of data streams that preserves all  the properties of the metric as computed on the streams in their entirety. 
%\end{itemize}
The main objective of this paper is to propose a novel  metric that reflects the relationships between any two  discrete probability distributions   in the context of massive data streams. Specifically, this metric, denoted by \emph{Sketch {$\star$}-metric},    allows us to efficiently  estimate a broad class of  distances measures between any two large data streams by computing these distances only using  compact synopses or sketches of the streams.    The  \emph{Sketch {$\star$}-metric} is distribution-free and makes no assumption about the underlying data volume. It is thus capable of comparing any two data streams, identifying their  correlation if any, and more generally, it allows us to acquire a deep understanding of the structure of the input streams.  Formalization of this metric is the first contribution of this paper. 

The interest  of estimating  distances  between any two data streams is important in  data intensive applications. Many different domains  are  concerned by such analyses  including machine learning, data mining, databases, information retrieval, and network monitoring.  In all these applications,  it is necessary  to quickly and precisely   process a huge amount of data. For instance, in IP network management, the analysis of input streams will  allow  to rapidly  detect the presence of outliers or intrusions when changes in the communication patterns occur~\cite{SKSZC03}. In sensors networks, such an analysis will enable the   correlation of  geographical and environmental informations~\cite{Busnel2008SOLIST-or-Ho,CHZ02}.  Actually, the problem of detecting changes or outliers in a data stream is similar to the problem of identifying patterns that do not conform to the expected behavior, which has been an active area of research for many decades. For instance, depending on the specificities of the domain considered
and the type of outliers considered, different methods have been designed, namely classification-based, clustering-based, nearest neighbor based, statistical, spectral, and information theory.
To accurately analyze  streams of data,  a panel of information-theoretic measures and  distances have been proposed to answer  the specificities of the  analyses. Among them, the most commonly used are  the  Kullback-Leibler (KL) divergence~\cite{KLDIVERGE}, the $f$-divergence introduced by 
Csiszar, Morimoto and Ali \& Silvey~\cite{csiszar63,JPSJ.18.328,ali-silvey}, the Jensen-Shannon divergence and the Battacharyya distance~\cite{MR0010358}.  More details can be found in the comprehensive survey of Basseville~\cite{basseville}. 

Unfortunately, computing information theoretic measures of distances in the data stream model is challenging essentially because one needs to process streams on the fly (\emph{i.e}, in one-pass), on huge amount of data, and by using very little storage with respect to the size of the stream. In addition the analysis must be robust over time to detect any sudden change in the observed streams (which may be the manifestation of routers deny of service attack or worm propagation).  
 We tackle this issue by presenting  an approximation algorithm that constructs a sketch of the stream from which  the \emph{Sketch {$\star$}-metric} is  computed. This algorithm is a one-pass algorithm. It uses very basic computations, and little storage space (\emph{i.e.}, $\mathcal O\left( t (\log n + k \log m)\right)$ where $k$ and $t$ are precision parameters, $m$ is an upper-bound of stream size and $n$ the number of distinct items in the stream). It does not need any information on the size of input streams nor on their  structure. This consists in the second contribution of the paper.  
 
% This contribution positively  answers an open question raised by Charikar in \cite{Similarity Estimation Techniques from Rounding
%Algorithms, stoc 2002}.  
Finally,  the   robustness of our approach is validated with a detailed experimentation study based on synthetic traces that range from stable streams to highly skewed ones.

The paper is organized as follows. First, Section~\ref{sec:relatedwork} reviews the related work on  classical generalized metrics and their applications on the data stream model while 
Section~\ref{sec:model} describes this model. Section~\ref{sec:KL} presents the necessary background that makes the paper self-contained. Section~\ref{toc:star-metric} formalizes the  \emph{Sketch {$\star$}-metric}. Section~\ref{toc:algo} presents the algorithm that fairly approximates the \emph{Sketch {$\star$}-metric} in one pass and Section~\ref{toc:eval} presents extensive experiments (on both synthetic traces and real data) of our algorithm. 
%In Section~\ref{toc:simu}, we empirically evaluate the accuracy of the estimation provided by AnKLe by comparing it to the exact value of the KL divergence on different data streams and also to adapted versions of state-of-the-art estimator-based algorithms, namely, Alon \emph{et al.}~\cite{AMS}  and Chakrabrti \emph{et al.}~\cite{CCM07}. 
Finally, we conclude in Section~\ref{sec:conclusion}.

%% file: related-works.tex
%!TEX root =  AB13-INFOCOM.tex

\section{Related Work}\label{sec:relatedwork}

Work on data stream analysis mainly focuses on efficient methods (data-structures and algorithms) to  answer different kind of queries over massive data streams. Mostly, these methods consist in deriving statistic  estimators over the data stream, in creating  summary representations of streams (to build histograms, wavelets, and quantiles), and in comparing data streams. Regarding the construction of estimators, a seminal work is due to  Alon \emph{et al.}~\cite{AMS}. The authors  have proposed estimators of the frequency moments $F_{k}$ of a stream,  which are important statistical tools  that allow to quantify specificities of a data stream. Subsequently, a lot of attention has been  paid to the strongly related notion of the entropy of a stream, and all notions based on entropy (\emph{i.e.}, norm and relative entropy)~\cite{coverthomas1991}. These notions are essentially related to the  quantification of the amount of randomness of a stream (\emph{e.g},~\cite{CCM07,GMcV06,CBM06,LSOXZ06,EDCC2012,NCA2012}). The construction of  synopses or sketches  of the data stream have been proposed for different applications (\emph{e.g},~\cite{CCF04,CM05,CG07}). 

Distance and divergence measures are key measures in statistical inference and data processing problems~\cite{basseville}. There exists two largely used broad classes of measures, namely the $f$-divergences and the Bregman divergences. Among them, there exists two classical distances, namely the Kullback-Leibler (KL) divergence and the Hellinger distance, that are very important to quantify the amount of information that separates two distributions. In~\cite{EDCC2012}, the authors have proposed a one pass algorithm for estimating the KL divergence of an observed stream compared to an expected one.  Experimental evaluations have shown that the estimation provided by this algorithm is accurate for different adversarial settings for which the quality of other methods dramatically decreases. However, this solution assumes that the expected stream is the uniform one, that is a fully random stream.  Actually in~\cite{GIG07}, the authors propose a characterization of the information divergences that are not sketchable. They have proven that any distance that has not ``norm-like'' properties is not sketchable. In the present paper, we go one step further by formalizing a metric that allows to efficiently and accurately estimate a broad class of  distances measures between any two large data streams by computing these distances uniquely on compact synopses or sketches of  streams.

%% file: model.tex
%!TEX root =  AB13-INFOCOM.tex

\section{Data Stream Model}\label{sec:model}

We consider a system in which a node $P$ receives a large data stream $\sigma=a_{1},a_{2},\ldots,a_{m}$ of data items that arrive sequentially. 
In the following, we describe a single instance of $P$, but clearly multiple instances of $P$ may  co-exist  in a system (\emph{e.g.}, in case $P$ represents a router, a base station in a sensor network). 
Each data item $i$ of the stream $\sigma$  is drawn from the universe $\Omega= \{1,2,\ldots,n\}$ where $n$ is very large.   Data items can be repeated multiple times in the stream.   In the following we suppose that  the length $m$ of the stream is not known.  Items in the stream arrive regularly and quickly, and due to memory constraints, need to be processed sequentially and in an online manner. Therefore, node $P$ can locally store only a small fraction of the items and perform simple operations on them. The  algorithms  we consider in this work are characterized by the fact that they can approximate some function on $\sigma$ with a very limited amount of memory. We refer the reader to~\cite{Muthukrishnan} for a detailed description of data streaming models and algorithms.  

%\textbf{REMARK: I remove the adversary as we do not consider it in this paper... But it can come back if you prefer...}
%\paragraph{Adversary Model} 
%We suppose that the adversary is omnipotent in the sense that  it may actively tamper with the data stream of any node by observing, inserting, dropping or re-ordering items of their input stream. 
%%Of course, we also assume that there is some limit to the number of items that the adversary can influence. For instance, some items of the stream will still be provided by honest nodes. 
%The activity of the adversary can be detected by an honest node provided that it can accurately estimate the divergence between the observed stream and the ideal one. The presence of such a divergence is important as it may be a good indicator of attacks. For instance, in large scale systems, it might be used as an alarm  to prevent the adversary from poisoning  routing tables  (also called eclipse attacks~\cite{sitmorris02}) by freezing routing tables updates as long as the relative entropy is too high.  We suppose that the algorithm used by a node to estimate the divergence is public knowledge (\emph{i.e.}, to avoid some kind of security by obscurity), however the adversary has not access to the local random coins used in the algorithm (if any). 
%

\section{Information Divergence of Data Streams}\label{sec:KL}
\label{sec:preliminaries}

We first present notations and background that make this paper self-contained.

\subsection{Preliminaries}
A natural approach to study a data stream $\sigma$ is to model it as an empirical data distribution over the universe $\Omega$, given by $(p_1,p_2,\ldots, p_n)$ with $p_i = x_i/ m$, and $x_{i}=|\{j:a_{j}=i\}|$ representing  the number of times data item $i$ appears in $\sigma$. We have $m = \sum_{i \in \Omega} x_{i}$. 
\subsubsection{Entropy} 
  Intuitively, the entropy is a measure of the randomness of a data stream $\sigma$. The entropy $H(\sigma)$ is minimum (\emph{i.e.}, equal to zero)  when all the items in the stream are the same,  and it reaches its maximum (\emph{i.e.}, $\log_{2} m$) when all the items in the stream are distinct. Specifically, we have 
$H(\sigma)=-\sum_{i\in \Omega} p_{i}\log_{2} p_{i}.$
The $\log$ is to the base 2 and thus entropy is expressed in bits. By convention, we have $0\log 0=0$. 
%Without loss of generality, we assume that the items are ordered so that $m_{1} \geq m_{2} \geq \ldots \geq m_{n}$.  
Note that the number of times $x_{i}$ item $i$ appears in a stream  is commonly called the frequency of item $i$. The norm of the entropy is defined as $F_{H}=\sum_{i\in\Omega}x_{i}\log x_{i}$.

%\subsubsection{Frequency moments}
%Frequency moments are important statistical tools that have been introduced by Alon \emph{et al.}~\cite{AMS}. Computing  frequency moments $F_{k}$ allows to  quantify the amount of skew in a data stream. For each $k\geq 0$, the $k$-th frequency moment $F_{k}$ of $\sigma$ is defined as 
%$F_{k}=\sum_{u \in N}m_{u}^{k},$
%where $m_{u}$ represents the number of occurrences of $u$ in the stream (\emph{c.f.} the definition of $m_{u}$ above). 
%Among the remarkable moments,  $F_{0}$ represents the number $n$ of distinct elements in a stream while $F_{1}$ corresponds to the size $m$ of the stream.    

\subsubsection{2-universal Hash Functions} 
In the following, we intensively use hash functions randomly picked from a $2$-universal hash family. A collection $\mathcal H$ of hash functions $h: \{1,\ldots, M\} \rightarrow \{0, \ldots, M'\}$ is said to be \emph{2-universal} if for every ${h\in \mathcal H}$ and for every two different items $i,j \in [M]$,
$\mathbb P \{h(i)=h(j)\}\leq \frac{1}{M'},$
which is exactly the probability of collision obtained if the hash function assigned truly random values to any $i\in[M]$. In the following, notation $[M]$ means $\{1,\ldots,M\}$.

%\subsubsection{Randomized~\mbox{$(\varepsilon,\delta)$-approximation Algorithm}}
%A randomized algorithm $\mathcal A$ is said to be an $(\varepsilon,\delta)$-approximation of a function $\phi$ on $\sigma$ if for any sequence of items in the input stream $\sigma$, $\mathcal A$ outputs $\hat{\phi}$ such that $\Pr\{\mid \hat{\phi} - \phi \mid> \varepsilon \phi \} < \delta$, where $\varepsilon, \delta >0$ are given as parameters of the algorithm. 

\subsection{Metrics and divergences}

\subsubsection{Metric definitions}

The classical definition of a metric is based on a set of four axioms. 
\begin{definition}[Metric]
Given a set $X$, a \textbf{metric} is a function $d : X \times X \rightarrow \mathbb R$ such that, for any $x,y,z\in X$, we have:
\begin{align}
\text{Non-negativity:} \quad & d(x, y) \geq 0 \label{eq:non-neg}\\
\text{Identity of indiscernibles:} \quad & d(x, y) = 0 \Leftrightarrow x = y\label{eq:identity}\\
\text{Symmetry:} \quad & d(x, y) = d(y,x) \label{eq:symmetry}\\
\text{Triangle inequality:} \quad & d(x, y) \leq d(x,z) + d(z,y) \label{eq:triangle-ineq}
\end{align}
\end{definition}

In the context of information divergence, usual distance functions  are not precisely metric. Indeed, most of divergence functions do not verify the 4 axioms, but only a subset of them. We recall  hereafter some definitions of generalized metrics.

\begin{definition}[Pseudometric]
Given a set $X$, a \textbf{pseudometric} is a function that verifies the axioms of a metric with the exception of the \emph{identity of indiscernible}, which is replaced by 
\[\forall x\in X, d(x,x) = 0.\]
\end{definition}

Note that this definition allows that $d(x,y) = 0$ for some $x\neq y$ in $X$.

\begin{definition}[Quasimetric]
Given a set $X$, a \textbf{quasimetric} is a function that verifies all the axioms of a metric with the exception of the \emph{symmetry} (\emph{cf.} Relation~\ref{eq:symmetry}).
\end{definition}

\begin{definition}[Semimetric]
Given a set $X$, a \textbf{semimetric} is a function that verifies all the axioms of a metric with the exception of the \emph{triangle inequality} (\emph{cf.} Relation~\ref{eq:triangle-ineq}).
\end{definition}

\begin{definition}[Premetric]
Given a set $X$, a \textbf{premetric} is a pseudometric  that relax both the \emph{symmetry}  and \emph{triangle inequality} axioms.
\end{definition}

\begin{definition}[Pseudoquasimetric]
Given a set $X$, a \textbf{pseudoquasimetric} is a function that relax both the \emph{identity of indiscernible}  and the \emph{symmetry} axioms.
\end{definition}

Note that the latter definition simply corresponds to a premetric satisfying the triangle inequality. Remark also that all the generalized metrics preserve the \emph{non-negativity} axiom.

\subsubsection{Divergences}

We now  give  the definition of two broad classes of generalized metrics, usually denoted as \emph{divergences}.

\paragraph{$f$-divergence} 

Mostly used in the context of statistics and probability theory, a $f$-divergence ${\mathcal D}_f$ is a premetric that guarantees monotonicity and convexity. 

\begin{definition}[$f$-divergence] \label{def:f-div}
Let $p$ and $q$ be two $\Omega$-point distributions. Given a convex function $f : (0, \infty)\rightarrow \mathbb R$ such that $f(1) = 0$, the \textbf{$f$-divergence} of $q$ from $p$ is:
\[{\mathcal D}_f(p || q) = \sum_{i\in\Omega} q_i f\left(\frac{p_i}{q_i}\right),\]
where by convention $0 f(\frac{0}{0}) = 0$, $a f(\frac{0}{a}) = a\lim_{u\rightarrow 0} f(u)$, and $0 f(\frac{a}{0}) = a\lim_{u\rightarrow \infty} f(u)/u$ if these limits exist.
\end{definition}

Following this definition, any $f$-divergence verifies both monotonicity and convexity.
\begin{property}[Monotonicity] \label{prop:monotonicity}
Given $\kappa$ an arbitrary transition probability that respectively transforms two $\Omega$-point distributions $p$ and $q$ into $p_\kappa$ and $q_\kappa$, we have:
\[{\mathcal D}_f(p || q) \geq {\mathcal D}_f(p_\kappa || q_\kappa).\] 
\end{property}

\begin{property}[Convexity] \label{prop:convexity}
Let $p_1$, $p_2$, $q_1$ and $q_2$ be four $\Omega$-point distributions. Given any $\lambda \in [0,1]$, we have:
\begin{align*}{\mathcal D}_f\left(\lambda p_1 + (1-\lambda)p_2 || \lambda q_1 + (1-\lambda) q_2\right) \\
\leq \lambda {\mathcal D}_f(p_1 || q_1) + (1-\lambda) {\mathcal D}_f(p_2 || q_2).
\end{align*} 
\end{property}

This class of divergences has been introduced in  independent works by chronologically Csisz\'ar, Morimoto and Ali \& Silvey~\cite{csiszar63,JPSJ.18.328,ali-silvey}. All the distance measures in the so-called \emph{Ali-Silvey distances} are applicable to quantifying statistical differences between data streams.

\paragraph{Bregman divergence}

Initially proposed in~\cite{Bregman1967200}, this class of generalized metrics encloses quasimetrics and  semimetrics, as these divergences do not satisfy the triangle inequality nor symmetry.
\begin{definition}[Bregman divergence]
Given  $F$ a continuously-differentiable and strictly convex function defined on a closed convex set $C$, the \textbf{Bregman divergence} associated with $F$ for $p,q\in C$ is defined as
\[{\mathcal B}_F (p || q) = F(p) - F(q) - \left\langle \nabla F (q) , (p - q)\right\rangle.\]
where the operator $\langle\cdot,\cdot\rangle$ denotes the inner product.
\end{definition}

In the context of data stream, it is possible to reformulate this definition according to probability theory. Specifically,
\begin{definition}[Decomposable Bregman divergence]\label{def:bregman}
Let $p$ and $q$ be two $\Omega$-point distributions. Given a strictly convex function $F : (0, 1]\rightarrow \mathbb R$, the \textbf{Bregman divergence} associated with $F$ of $q$ from $p$ is defined as
\[{\mathcal B}_F (p || q) = \sum_{i\in\Omega} \left(F(p_i) - F(q_i) - (p_i - q_i) F' (q_i) \right).\]
\end{definition}

Following these definitions, any Bregman divergence verifies non-negativity and convexity in its first argument, but not necessarily in the second argument. Another interesting property is given by thinking of the Bregman divergences as an operator of the function $F$.
\begin{property}[Linearity] \label{prop:linearity}
Let $F_1$ and $F_2$ be two strictly convex and differentiable functions. Given any $\lambda \in [0,1]$, we have that
\[{\mathcal B}_{F_1+\lambda F_2} (p || q) = {\mathcal B}_{F_1} (p || q)  + \lambda {\mathcal B}_{F_2} (p || q) .\] 
\end{property}

\subsubsection{Classical metrics}

In this section, we present several commonly used metrics in $\Omega$-point distribution context. These specific metrics are used in the evaluation part presented in Section~\ref{toc:eval}.

\paragraph{Kullback-Leibler divergence}

The Kullback-Leibler (KL) divergence~\cite{KLDIVERGE}, also called the relative entropy, is a robust metric for measuring the statistical difference between two data streams. The KL divergence owns the special feature that it is both a $f$-divergence and a Bregman one (with $f(t) = F(t) = t \log t$).

 Given $p$ and $q$ two $\Omega$-point distributions, the Kullback-Leibler divergence is then defined as 
\begin{equation}\label{eq:kld}
\mathcal D_{KL}(p || q)  = \sum_{i\in\Omega} p_{i}\log\frac{p_{i}}{q_{i}}=H(p,q)-H(p),
\end{equation}
where $H(p)=-\sum p_{i}\log p_{i}$ is the (empirical) entropy of $p$ and $H(p,q)=-\sum p_{i}\log q_{i}$ is the cross entropy of $p$ and $q$.   
%As we use a logarithm in base $2$, the divergence is measured in bits. When $p_{n}=q_{n}$, the KL divergence is minimal and is equal to zero. Let $p^{(\mathcal U)}$ be the uniform distribution corresponding to a uniform stream (\emph{i.e.},~$\forall u\in \sigma, p^{(\mathcal U)}_u= \frac{1}{n}$),  and $q$ be the probability distribution corresponding to the input stream. In the rest of this paper and according to the classical use of the KL-divergence, we consider $\mathcal D(q||p^{(\mathcal U)})$ as a measure of the divergence of the current stream from the ideal one.  
% While all the distance measures in the Ali-Silvey distances are applicable to quantifying statistical differences between data streams, the KL divergence is particularly suited to our context since it gives rise to a small number of false positives when the two data streams are not significantly different. 

\paragraph{Jensen-Shannon divergence}
The Jensen-Shannon divergence (JS)  is a symmetrized and smoothed version of the Kullback-Leibler divergence. Also known as information radius (IRad) or total divergence to the average, it is defined as
\begin{equation}\label{eq:jsd}
\mathcal D_{JS}(p || q)  = \frac{1}{2} D_{KL}(p || \ell) + \frac{1}{2} D_{KL}(q || \ell),
\end{equation}
where $\ell=\frac{1}{2} (p + q)$. Note that the square root of this divergence is a metric.

\paragraph{Bhattacharyya distance} 
The Bhattacharyya distance is derived from his proposed measure of similarity between two multinomial distributions, known as the Bhattacharya coefficient (BC)~\cite{MR0010358}. It is defined as
\begin{align}\label{eq:bd}
\mathcal D_{B}(p || q)  = -\log (BC(p,q)) \text{ where } BC(p,q) = \sum_{i\in\Omega} \sqrt{{p_{i}}{q_{i}}}.\nonumber
\end{align}
This distance is a semimetric as it does not verify the triangle inequality. Note that the famous Hellinger distance~\cite{hellinger09} equal to  $\sqrt{1-BC(p,q)}$ verifies it.

%% file: star-metric.tex
%!TEX root =  AB13-INFOCOM.tex

\section{Sketch {\Large $\star$}-metric}\label{toc:star-metric}

We now present a method to sketch two input data streams $\sigma_{1}$ and $\sigma_{2}$, and to compute any generalized metric $\phi$ between these sketches  such that this computation preserves all the properties of  $\phi$  computed on $\sigma_{1}$ and $\sigma_{2}$.  Proof of correctness of this method  is presented in this section. %We prove in the following  that any  metric has  its axioms and properties preserved  on the sketch.  %This generalized metric is denoted \emph{Sketch $\star$-metric} in the following.

 \begin{definition}[Sketch $\star$-metric]\label{def:sketch}
Let $p$ and $q$ be any two $n$-point distributions. Given a precision parameter $k$, and any generalized metric $\phi$ on the set of all $\Omega$-point distributions,  there exists  a  \textbf{Sketch $\star$-metric $\widehat{\phi} _{k}$}  defined as follows
\[\widehat{\phi} _{k} (p || q) = \max_{\rho \in \mathcal P_{k}(\Omega)} \phi (\widehat{p} _{\rho} || \widehat{q} _{\rho}) \mbox{ \textrm{with} } \forall a \in \rho, \; \widehat{p}_{\rho}(a) = \sum_{i\in a} p(i),\]
where $\mathcal P_{k}(\Omega)$ is the set of all partitions of an $\Omega$-element set into exactly $k$ nonempty and mutually exclusive cells. 
 \end{definition}
 
 \begin{remark} Note that for $k > \Omega$, it does not exist a partition of $\Omega$ into $k$ nonempty parts. By convention, we consider that $\widehat{\phi} _{k} (p || q) = {\phi}(p || q)$ in this specific context.
 \end{remark}
 
 In this section, we focus on the preservation of axioms and properties of a  generalized metric $\phi$ by the corresponding \emph{Sketch $\star$-metric} $\widehat{\phi}_k$.
 
 \subsection{Axioms preserving}
 
 \begin{theorem} \label{thm:metric}
 Given any generalized metric $\phi$ then, for any $k\in\Omega$, the corresponding Sketch $\star$-metric $\widehat{\phi}_k$ preserves all the axioms of $\phi$.
 \end{theorem}
 \begin{IEEEproof}
 The proof is directly derived from Lemmata~\ref{lem:non-neg}, \ref{lem:identity}, \ref{lem:symm} and~\ref{lem:triangle}.
 \end{IEEEproof}

 \begin{lemma}[Non-negativity]\label{lem:non-neg}
  Given any generalized metric $\phi$  verifying the Non-negativity axiom then, for any $k\in\Omega$, the corresponding Sketch $\star$-metric $\widehat{\phi} _k$ preserves  the Non-negativity axiom.
 \end{lemma}
 
 \begin{IEEEproof}
Let $p$ and $q$ be any  two $\Omega$-point distributions.  By definition, 
 \[
 \widehat{\phi} _{k} (p || q) =  \max_{\rho \in \mathcal P_{k}(\Omega)} \phi (\widehat{p} _{\rho} || \widehat{q} _{\rho})
\]
 As for any two $k$-point distributions,  $\phi$ is positive we have $ \widehat{\phi} _{k} (p || q) \geq 0$
 that concludes the proof.
 \end{IEEEproof}

  \begin{lemma}[Identity of indiscernible]\label{lem:identity}
  Given any generalized metric $\phi$ verifying the Identity of indiscernible axiom then, for any $k\in \Omega$, the corresponding Sketch $\star$-metric $\widehat{\phi}_k$ preserves  the Identity of indiscernible axiom.
 \end{lemma}
 
  \begin{IEEEproof}
Let $p$ be any $\Omega$-point distribution. We have 
\[\widehat{\phi} _{k} (p || p) = \max_{\rho \in \mathcal P_{k}(\Omega)} \phi (\widehat{p} _{\rho} || \widehat{p} _{\rho}) = 0,\]
due to   $\phi$  Identity of indiscernible axiom.

Consider now two $\Omega$-point distributions $p$ and $q$ such that $\widehat{\phi} _{k} (p || q) = 0$. Metric  $\phi$ verifies both the non-negativity  axiom (by construction) and the Identity of indiscernible axiom (by assumption). Thus we have
$\forall \rho \in \mathcal P_{k}(\Omega), \widehat{p}_{\rho} = \widehat{q}_{\rho},$
leading to 
\begin{equation}\label{eq:lem:iden}
\forall \rho \in \mathcal P_{k}(\Omega), \forall a \in \rho, \sum_{i\in a} p(i) = \sum_{i\in a} q(i).\end{equation}
Moreover, for any $i\in\Omega$, there exists a partition $\rho\in \mathcal P_{k}(\Omega)$ such that $\{i\}\in\rho$. By  Equation~\ref{eq:lem:iden},  $\forall i \in\Omega, p(i)=q(i)$, and so $p=q$.

Combining the two parts of the proof leads to
$ \widehat{\phi} _{k} (p || q) = 0 \Longleftrightarrow p = q,
 $
 which  concludes the proof of the Lemma.
 \end{IEEEproof}

  \begin{lemma}[Symmetry]\label{lem:symm}
  Given any generalized metric $\phi$ verifying  the Symmetry axiom then, for any $k\in\Omega$, the corresponding Sketch $\star$-metric $\widehat{\phi}_k$ preserves  the Symmetry axiom.
 \end{lemma}

 \begin{IEEEproof}
Let $p$ and $q$ be any two $\Omega$-point distributions.  We have
 \[
 \widehat{\phi} _{k} (p || q) =  \max_{\rho \in \mathcal P_{k}(\Omega)} \phi (\widehat{p} _{\rho} || \widehat{q} _{\rho}).
\]
 Let $\overline \rho \in \mathcal P_{k}(\Omega)$  be a $k$-cell partition such that $\phi (\widehat{p} _{\overline\rho} || \widehat{q} _{\overline\rho})=  \max_{\rho \in \mathcal P_{k}(\Omega)} \phi (\widehat{p} _{\rho} || \widehat{q} _{\rho})$. We get
 \[
 \widehat{\phi} _{k} (p || q) =  \phi (\widehat{p} _{\overline\rho} || \widehat{q} _{\overline\rho}) = \phi (\widehat{q} _{\overline\rho} || \widehat{p} _{\overline\rho}) \leq  \widehat{\phi} _{k} (q || p).
\]
 By symmetry, considering $\underline \rho \in \mathcal P_{k}(\Omega)$ such that $\phi (\widehat{q} _{\underline\rho} || \widehat{p} _{\underline\rho})=  \max_{\rho \in \mathcal P_{k}(\Omega)} \phi (\widehat{q} _{\rho} || \widehat{p} _{\rho})$, we also have $\widehat{\phi} _{k} (q || p) \leq \widehat{\phi} _{k} (p || q)$, 
 which concludes the proof.
 \end{IEEEproof}

  \begin{lemma}[Triangle inequality]\label{lem:triangle}
  Given any generalized metric $\phi$ verifying the Triangle inequality axiom then, for any $k\in \Omega$, the corresponding Sketch $\star$-metric $\widehat{\phi}_k$ preserves  the Triangle inequality axiom.
 \end{lemma}

 \begin{IEEEproof}
Let $p$, $q$ and $r$ be any  three $\Omega$-point distributions. Let $\overline \rho \in \mathcal P_{k}(\Omega)$  be a $k$-cell partition such that $\phi (\widehat{p} _{\overline\rho} || \widehat{q} _{\overline\rho})=  \max_{\rho \in \mathcal P_{k}(\Omega)} \phi (\widehat{p} _{\rho} || \widehat{q} _{\rho})$.

 We have
 \begin{align*}
 \widehat{\phi} _{k} (p || q) ~& =  \phi (\widehat{p} _{\overline\rho} || \widehat{q} _{\overline\rho})\\
 & \leq \phi (\widehat{p} _{\overline\rho} || \overline{r}_{\overline\rho}) + \phi (\overline{r}_{\overline\rho} || \widehat{q} _{\overline\rho})\\
 & \leq \max_{\rho \in \mathcal P_{k}(\Omega)} \phi (\widehat{p} _{\rho} || \overline{r}_{\rho}) + \max_{\rho \in \mathcal P_{k}(\Omega)} \phi (\overline{r}_{\rho} || \widehat{q} _{\rho})\\
 & = \widehat{\phi} _{k} (p || r) + \widehat{\phi} _{k} (r || q)
\end{align*}
 that concludes the proof.
  \end{IEEEproof}

 \subsection{Properties preserving}

 \begin{theorem}\label{thm:fdiv}
Given a $f$-divergence $\phi$ then, for any $k\in \Omega$, the corresponding Sketch $\star$-metric $\widehat{\phi}_k$ is also a $f$-divergence.
 \end{theorem}
 
 \begin{IEEEproof}
 From Theorem~\ref{thm:metric}, $\widehat{\phi}_k$ preserves the axioms of the generalized metric. Thus, $\widehat{\phi}_k$ and $\phi$ are in the same equivalence class. Moreover, from Lemma~\ref{lem:monotonicity}, $\widehat{\phi}_k$ verifies the monotonicity property. Thus, as the $f$-divergence is the only class of decomposable information \emph{monotonic} divergences (\emph{cf.}~\cite{Csiszar1978}), $\widehat{\phi}_k$ is also a  $f$-divergence.
 \end{IEEEproof}
 
 \begin{theorem}\label{thm:bregman}
Given a Bregman divergence $\phi$ then, for any $k\in \Omega$, the corresponding Sketch $\star$-metric $\widehat{\phi}_k$ is also a Bregman divergence.
 \end{theorem}
 
 \begin{IEEEproof}
 From Theorem~\ref{thm:metric}, $\widehat{\phi}_k$ preserves the axioms of the generalized metric. Thus, $\widehat{\phi}_k$ and $\phi$ are in the same equivalence class. Moreover, the  Bregman divergence is characterized by the property of transitivity  (\emph{cf.}~\cite{Csiszar91}) defined as follows.  Given $p$, $q$ and $r$ three  $\Omega$-point distributions such that $q = \Pi(L|r)$ and $p\in L$, with $\Pi$ is a selection rule according to the definition of Csisz\'ar in~\cite{Csiszar91} and $L$ is a subset of the $\Omega$-point distributions, we have the Generalized Pythagorean Theorem:
 \[
 \phi (p || q) + \phi (q || r) = \phi (p || r).
 \]
 Moreover the authors in~\cite{Amari10} show that the set $S_n$ of all discrete probability distributions over $n$ elements $(X = \{x_1, \ldots , x_n\})$ is a Riemannian manifold, and it owns another different dually flat  affine structure. They also show that these dual structures give rise to the generalized Pythagorean theorem. This is verified for the coordinates in $S_n$ and for the dual coordinates~\cite{Amari10}. Combining these results with the projection theorem~\cite{Csiszar91,Amari10}, we obtain that
\begin{align*}
 \widehat{\phi}_k (p || r) & = \max _{\rho \in \mathcal P_{k}(n)} \phi (\widehat{p} _{\rho} || \widehat{r}_{\rho})\\
&  = \max _{\rho \in \mathcal P_{k}(n)} \left(\phi (\widehat{p} _{\rho} || \widehat{q} _{\rho}) + \phi (\widehat{q} _{\rho} || \widehat{r}_{\rho})\right)\\
&  = \max _{\rho \in \mathcal P_{k}(n)} \phi (\widehat{p} _{\rho} || \widehat{q} _{\rho}) + \max _{\rho \in \mathcal P_{k}(n)} \phi (\widehat{q} _{\rho} || \widehat{r}_{\rho})\\
&  =  \widehat{\phi}_k (p || q) +  \widehat{\phi}_k (q || r)
\end{align*}
Finally, by  the characterization of Bregman divergence through transitivity~\cite{Csiszar91}, and reinforced with Lemma~\ref{lem:linearity} statement, $\widehat{\phi}_k$ is also a Bregman divergence.
 \end{IEEEproof}

  In the following, we show that the \emph{Sketch $\star$-metric} preserves the properties of divergences.

   \begin{lemma}[Monotonicity]\label{lem:monotonicity}
  Given any generalized metric $\phi$  verifying the Monotonicity property then, for any $k\in\Omega$, the corresponding Sketch $\star$-metric $\widehat{\phi}_k$ preserves  the Monotonicity property.
 \end{lemma}
 
  \begin{IEEEproof}
Let $p$ and $q$ be any two $\Omega$-point distributions.  Given $c< N$, consider a partition $\mu\in \mathcal P_{c}(\Omega)$. As $\phi$ is monotonic, we have $\phi (p || q) \geq \phi (\widehat{p} _{\mu} || \widehat{q} _{\mu})$~\cite{Amari09}. We  split the proof into two cases:

Case (1). Suppose that $c \geq k$. Computing $\widehat{\phi} _{k} (\widehat{p} _{\mu} || \widehat{q} _{\mu})$ amounts in considering only  the $k$-cell partitions $\rho\in \mathcal P_{k}(\Omega)$ that verify
\[\forall b\in\mu, \exists a\in \rho, b\subseteq a.\]
These partitions form a subset of $\mathcal P_{k}(\Omega)$. The maximal value of $\phi (\widehat{p} _{\rho} || \widehat{q} _{\rho})$ over this subset cannot be greater than the maximal value over the whole $\mathcal P_{k}(\Omega)$. Thus we have
 \[
 \widehat{\phi} _{k} (p || q) =  \max_{\rho \in \mathcal P_{k}(\Omega)} \phi (\widehat{p} _{\rho} || \widehat{q} _{\rho}) \geq \widehat{\phi} _{k} (\widehat{p} _{\mu} || \widehat{q} _{\mu}).
\]

Case (2).  Suppose now that $c < k$. By definition, we have $\widehat{\phi} _{k}  (\widehat{p} _{\mu} || \widehat{q} _{\mu}) =  \phi (\widehat{p} _{\mu} || \widehat{q} _{\mu})$. Consider $\rho' \in \mathcal P_{k}(\Omega)$ such that $\forall b\in\mu, \exists a\in \rho', b\subseteq a$. It then exists a transition probability that respectively transforms $\widehat{p} _{\rho'}$ and $\widehat{q} _{\rho'}$ into $\widehat{p} _{\mu}$ and $\widehat{q} _{\mu}$.  As $\phi$ is monotonic, we have 
 \begin{align*}
\widehat{\phi} _{k} (p || q) & = \max_{\rho \in \mathcal P_{k}(\Omega)} \phi (\widehat{p} _{\rho} || \widehat{q} _{\rho}) \\
& \geq \phi (\widehat{p} _{\rho'} || \widehat{q} _{\rho'}) \\
& \geq \phi (\widehat{p} _{\mu} || \widehat{q} _{\mu}) = \widehat{\phi} _{k}  (\widehat{p} _{\mu} || \widehat{q} _{\mu}).
\end{align*}

Finally for any value of $c$, $\widehat{\phi} _{k}$ guarantees the monotonicity property. This concludes the proof.
 \end{IEEEproof}

  \begin{lemma}[Convexity]\label{lem:convexity}
  Given any generalized metric $\phi$ verifying  the Convexity property then, for any $k\in \Omega$, the corresponding Sketch $\star$-metric $\widehat{\phi}_k$ preserves  the Convexity property.
 \end{lemma}
 
  \begin{IEEEproof}
Let $p_1$, $p_2$, $q_1$ and $q_2$ be any four $\Omega$-point distributions. Given any $\lambda\in [0,1]$, we have:
 \begin{align*}
 & \widehat{\phi} _{k}\left(\lambda p_1 + (1-\lambda)p_2 || \lambda q_1 + (1-\lambda) q_2\right) \\
 & =  \max_{\rho \in \mathcal P_{k}(\Omega)} \phi \left(\lambda \widehat{p_1}_{\rho} + (1-\lambda)\widehat{p_2}_{\rho} || \lambda \widehat{q_1}_{\rho}+ (1-\lambda) \widehat{q_2}_{\rho}\right)
\end{align*}
Let $\rho' \in \mathcal P_{k}(\Omega)$ such that 
\begin{align*}
& \phi \left(\lambda \widehat{p_1}_{\rho'} + (1-\lambda)\widehat{p_2}_{\rho'} || \lambda \widehat{q_1}_{\rho'} + (1-\lambda) \widehat{q_2}_{\rho'}\right) \\
 & = \max_{\rho \in \mathcal P_{k}(\Omega)} \phi \left(\lambda \widehat{p_1}_{\rho} + (1-\lambda)\widehat{p_2}_{\rho} || \lambda \widehat{q_1}_{\rho} + (1-\lambda) \widehat{q_2}_{\rho}\right).
 \end{align*}
As $\phi$ verifies the Convex property, we have:
  \begin{align*}
 & \widehat{\phi} _{k}\left(\lambda p_1 + (1-\lambda)p_2 || \lambda q_1 + (1-\lambda) q_2\right) \\
 & = \phi \left(\lambda \widehat{p_1}_{\rho'} + (1-\lambda)\widehat{p_2}_{\rho'} || \lambda \widehat{q_1}_{\rho'} + (1-\lambda) \widehat{q_2}_{\rho'}\right) \\
 & \leq \lambda \phi (\widehat{p_1}_{\rho'} || \widehat{q_1}_{\rho'}) + (1-\lambda) \phi(\widehat{p_2}_{\rho'} || \widehat{q_2}_{\rho'})\\
&\leq \lambda\left(\max_{\rho \in \mathcal P_{k}(\Omega)} \phi (\widehat{p_1}_{\rho} || \widehat{q_1}_{\rho})\right) + (1-\lambda) \left(\max_{\rho \in \mathcal P_{k}(\Omega)} \phi(\widehat{p_2}_{\rho} || \widehat{q_2}_{\rho})\right)\\
& = \lambda \widehat{\phi} _{k} ({p_1} || {q_1}) + (1-\lambda) \widehat{\phi} _{k}({p_2} || {q_2})
 \end{align*}
that concludes the proof. 
 \end{IEEEproof}

 \begin{lemma}[Linearity]\label{lem:linearity}
The Sketch $\star$-metric definition preserves the Linearity property.
 \end{lemma}
 
  \begin{IEEEproof}
  Let $F_1$ and $F_2$ be  two strictly convex and differentiable functions, and any $\lambda\in [0,1]$.
  Consider the three Bregman divergences generated respectively from $F_1$, $F_2$ and $F_1+\lambda F_2$.
  
  Let $p$ and $q$ be two $\Omega$-point distributions.  We have:
 \begin{align*}
 \widehat{\mathcal B}_{{F_1+\lambda F_2}_{k}} (p || q) & =  \max_{\rho \in \mathcal P_{k}(\Omega)} {\mathcal B}_{F_1+\lambda F_2} (\widehat{p} _{\rho} || \widehat{q} _{\rho})\\
 & = \max_{\rho \in \mathcal P_{k}(n)} \left({\mathcal B}_{F_1} (\widehat{p} _{\rho} || \widehat{q} _{\rho}) + \lambda {\mathcal B}_{F_2} (\widehat{p} _{\rho} || \widehat{q} _{\rho})\right)\\
 & \leq  \widehat{\mathcal B}_{{F_1}_{k}} (p || q) +\lambda  \widehat{\mathcal B}_{{F_2}_{k}} (p || q)
\end{align*}

As $F_1$ and $F_2$ two strictly convex functions, and taken a leaf out of the Jensen's inequality,  we have:
 \begin{align*}
& \widehat{\mathcal B}_{{F_1}_{k}} (p || q) +\lambda \widehat{\mathcal B}_{{F_2}_{k}} (p || q) \\
& \leq \max_{\rho \in \mathcal P_{k}(\Omega)} \left({\mathcal B}_{F_1} (\widehat{p} _{\rho} || \widehat{q} _{\rho}) + \lambda {\mathcal B}_{F_2} (\widehat{p} _{\rho} || \widehat{q} _{\rho})\right)\\
 & =  \widehat{\mathcal B}_{{F_1+\lambda F_2}_{k}} (p || q) 
\end{align*}
that conclude the proof.
%\textbf{TODO: PROVE THE LAST RELATION !}
 \end{IEEEproof}

 This concludes the proof that the \emph{Sketch $\star$-metric} preserves all the axioms of a metric as well as the properties of $f$-divergences and Bregman divergences. We now show how to efficiently implement such a metric.

%% file: algorithm.tex
%!TEX root =  AB13-INFOCOM.tex

\section{Approximation algorithm}\label{toc:algo}

In this section, we propose an algorithm that computes the \emph{Sketch $\star$-metric} in one pass on the stream. 
By definition of the metric (\emph{cf.}, Definition~\ref{def:sketch}), we need to generate all the possible $k$-cell partitions. The number of these partitions follows the Stirling numbers of the second kind, which is equal to
$S(n,k) = \frac{1}{k!}\sum_{j=0}^k (-1)^{k-j}\binom {k} {j}j^n,$ where $n$ is the size of the items  universe.
Therefore, $S(n,k)$ grows exponentially with $n$. Unfortunately,  $n$ is very large. As the generating  function of $S(n,k)$ is equivalent to $x^n$, it is unreasonable in term of space complexity. We show in the following that generating $t=\lceil \log (1/\delta) \rceil$ random  $k$-cell partitions, where $\delta$ is the probability of error of our randomized algorithm, is sufficient to guarantee good overall performance of our metric.

%\textbf{TODO: INTRODUIRE L'ALGO DE COUNT MIN ET PRESENTER LE CALCUL EN FONCTION DE LA MATRICE.}

Our algorithm is inspired from the Count-Min Sketch algorithm proposed in~\cite{CM05} by Cormode and Muthukrishnan. Specifically, the Count-Min  algorithm is an  ($\varepsilon,\delta$)-approxi\-mation  algorithm that solves the \emph{frequency-estimation} problem.  For any items in the input stream  $\sigma$, the algorithm  outputs an estimation $\hat{f_{v}}$ of the frequency of  item $v$ such that $\mathbb P\{|\hat{f_{v}} - f_{v} |> \varepsilon f_{v} \} < \delta$, where $\varepsilon, \delta >0$ are given as parameters of the algorithm.   The estimation is computed by maintaining a two-dimensional array $C$ of $t \times k$ counters, and by using $t$ 2-universal hash functions $h_{i}$ ($1 \leq i \leq t$), where $k=2/\varepsilon$ and $t = \lceil \log(1/\delta) \rceil$. Each time an item $v$ is read from the input stream, this causes one counter of each line to be incremented, \emph{i.e.},  $C[h_{i}(v)]$ is incremented by one for each $i \in [1..t]$.  %When a query is issued to get an estimate $\hat{f_{v}}$  of the frequency of $v$, the returned value is the minimum among the $t$ lines of  $C[h_{i}(v)]$ ($1 \leq i \leq t$). 

To compute the \emph{Sketch $\star$-metric} of two streams $\sigma_{1}$ and $\sigma_{2}$, two sketches $\widehat{\sigma}_{1}$ and $\widehat{\sigma}_{2}$  of these streams are constructed according to the  above description. Note that there is no particular assumption on the  length of both streams $\sigma_{1}$ and $\sigma_{2}$. That is  their respective length is finite but unknown. %On the other hand, the support of the distributions is  the same.   
By construction of the 2-universal hash functions $h_{i}$ ($1 \leq i \leq t$), each line of  $C_{\sigma_{1}}$ and $C_{\sigma_{1}}$ corresponds to one  partition $\rho_{i}$ of the $N$-point empirical distributions of both  $\sigma_{1}$ and $\sigma_{2}$.  Thus when a query is issued to compute the given  distance $\phi$ between these two streams, the maximal value over all the $t$ partitions $\rho_{i}$ of the distance $\phi$  between  $\widehat{\sigma}_{1_{\rho_{i}}}$ and $\widehat{\sigma}_{2_{\rho_{i}}}$ is returned, \emph{i.e.}, the distance $\phi$ applied to the $i^{th}$ lines of $C_{\sigma_{1}}$ and $C_{\sigma_{1}}$ for $1\leq i \leq t$. Figure~\ref{algo:sketch} presents the pseudo-code of our algorithm.

\begin{algorithm}[t]
\caption{Sketch $\star$-metric algorithm\label{algo:sketch}}
\KwIn{Two input streams $\sigma_{1}$ and $\sigma_{2}$;  the distance $\phi$, $k$ and $t$ settings\;}
\KwOut{The distance $\overline{\phi}$ between  ${\sigma}_{1}$ and ${\sigma}_{2}$}
%\SetKw{KwAnd}{and}
%\KwData{}
%\KwResult{an estimation of $\mathcal D(q_\sigma || p^{(\mathcal U)})$}
%$\delta_0 \leftarrow -1$\;
%\KwData{
Choose $t$ functions $h: [n] \rightarrow [k]$, each from a 2-universal hash function family\;
%}
%\KwData{
%}
%\KwData{
$C_{\sigma_{1}}[1...t][1...k] \leftarrow 0$\;
$C_{\sigma_{2}}[1...t][1...k] \leftarrow 0$\;
%}

\For{$a_{j}\in \sigma_{1}$}{
	$v=a_{j}$\;
	\For{$i=1$ \KwTo $t$}{
		$C_{\sigma_{1}}[i][h_{i}(v)] \leftarrow C_{\sigma_{1}}[i][h_{i}(v)]+1$\;
	}
	}
\For{$a_{j}\in \sigma_{2}$}{
	$w=a_{j}$\;
	\For{$i=1$ \KwTo $t$}{
		$C_{\sigma_{2}}[i][h_{i}(w)] \leftarrow C_{\sigma_{2}}[i][h_{i}(w)]+1$\;
	}
	}	
On query $\overline{\phi}_k(\sigma_{1} || \sigma_{2})$ \Return $\overline{\phi}= \rm{max}_{1\leq i \leq t} \phi(C_{\sigma_{1}}[i][-],C_{\sigma_{2}}[i][-])$\;
\end{algorithm}

\begin{lemma}
Given parameters $k$ and $t$, Algorithm~\ref{algo:sketch} gives an approximation of the Sketch $\star$-metric, using  $\mathcal O\left( t (\log n + k \log m)\right)$ bits of space.
\end{lemma}

\begin{proof}
The matrices $C_{\sigma_i}$, for any $i \in \{1,2\}$,  are composed of $t\times k$ counters, which uses $\mathcal O\left(  \log m\right)$. On the other hand, with a suitable choice of hash family, we can store the hash functions above in $\mathcal O(t \log n)$ space.
\end{proof}

%% file: evaluation.tex
%!TEX root =  AB13-INFOCOM-RR.tex

\section{Performance Evaluation}\label{toc:eval}

We have implemented our \emph{Sketch $\star$-metric} and have conducted a series of experiments  on different types of streams and for different parameters settings.  We have  fed our algorithm with both real-world data  and synthetic traces. Real data give a realistic representation of some real systems, while the latter ones allow to capture  phenomenon which may be difficult to obtain from real-world traces, and thus allow to check the robustness of our metric. We have varied all the significant parameters of our algorithm, that is,  the maximal number of distinct data items $n$ in each stream, the number of cells $k$ of each generated partition, and the number of generated partitions $t$. 
For each parameters setting, we have conducted and averaged $100$ trials of the same experiment, leading to  a total of more than $300,000$ experiments for the evaluation of our metric. Real data have been downloaded from the repository of Internet network traffic~\cite{ITA}. We have used five traces among the available ones. 
Two of them represent two weeks logs of HTTP requests to the Internet service provider ClarkNet WWW server -- ClarkNet is a full Internet access provider for the Metro Baltimore-Washington DC area --  the other two ones contain two months of HTTP requests to the NASA Kennedy Space Center WWW server, and the last one represents seven months of HTTP requests to the WWW server of the University of Saskatchewan, Canada. In the following these data sets will be respectively referred to as ClarkNet, NASA, and Saskatchewan traces. Table~\ref{tab:traces} presents the statistics of these data traces, in term of stream size (\emph{cf.} ``\# items'' in the table), number of distinct items in each stream (\emph{cf.} ``\# distinct items'') and the number of occurrences of the most frequent item (\emph{cf.} ``max. freq.''). For more information on these data traces, an extensive analysis is available in~\cite{Arlitt96}.
We have evaluated the accuracy of our metric by comparing for each data set (real and synthetic), the results obtained with our algorithm  on the stream sketches (referred to as \emph{Sketch} in the legend) and the ones obtained on full streams (referred to as \emph{Ref} distance in the legend). That is, for each couple of input streams, and for each generalized metric $\phi$, we have computed both the exact distance between the two streams and the one as generated by our metric.
We now present the main lessons drawn from these experiments. 

\begin{table}
\centering
\begin{tabular}{lccc}
\hline
Data trace & \# items & \# distinct items & max. freq. \\
\hline
\hline
NASA (July) & 1,891,715 & 81,983 & 17,572\\
\hline
NASA (August) & 1,569,898 & 75,058 & 6,530\\
\hline
ClarkNet (August) & 1,654,929 & 90,516 & 6,075\\
\hline
ClarkNet (September) & 1,673,794 & 94,787 & 7,239\\
\hline
Saskatchewan & 2,408,625 & 162,523 & 52,695\\
\hline
\end{tabular}
\caption{Statistics of real data traces.}\label{tab:traces}
\vspace*{-0.5cm}
\end{table}

\begin{figure}
\centering
\includegraphics[width=.45\textwidth]{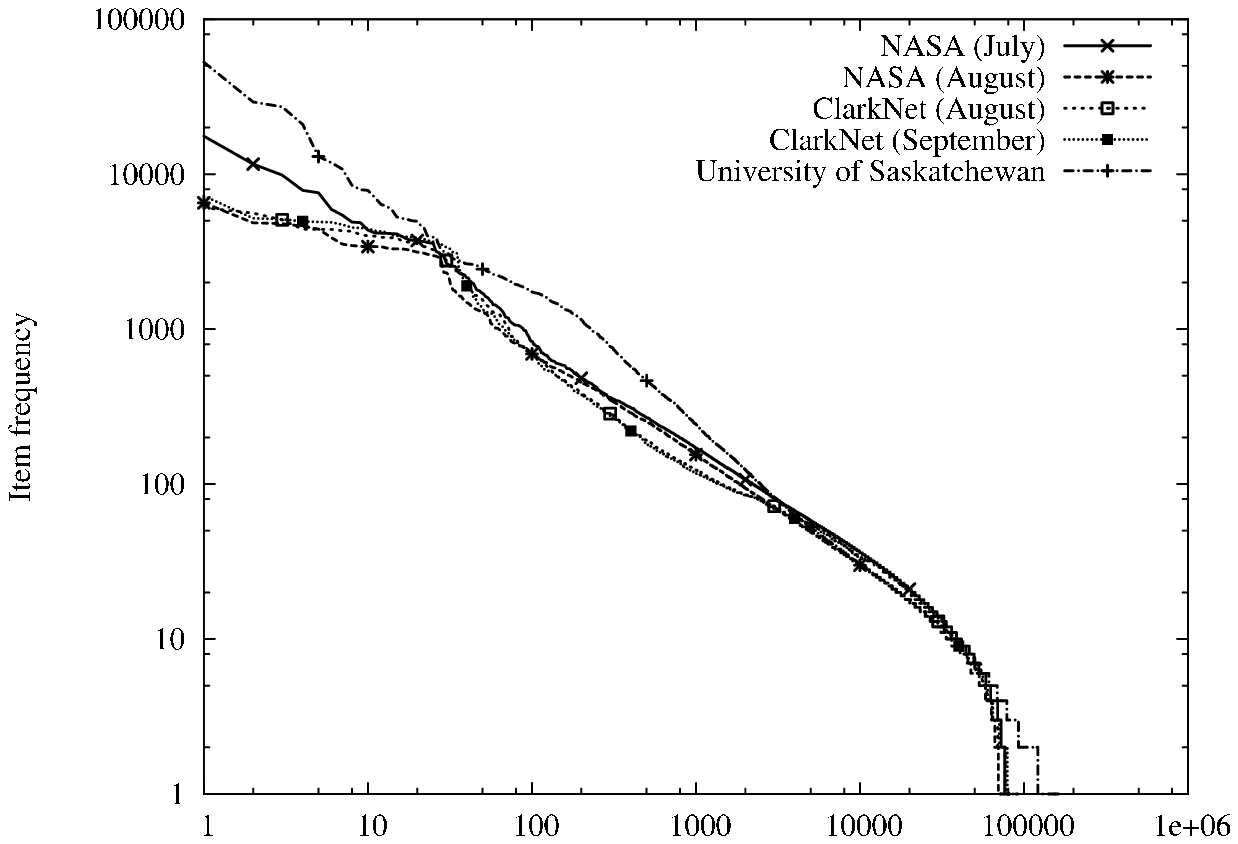}
\caption{Logscale distribution of frequencies for each real data trace.}
\label{fig:distrib-real}
\end{figure}

Figure~\ref{fig:accuracy} and~\ref{fig:accuracy-real} shows the accuracy of our metric as a function of the different input streams and the different generalized metrics applied on these streams.  All the histograms shown in Figures~\ref{fig:test-uniform-woCCM}--\ref{fig:test-zipf4-CCM-real} share the same legend, but for readability reasons, this legend is only indicated on histogram~\ref{fig:test-zipf1}.  Three generalized metrics have been used, namely  the Bhattacharyya distance, the Kullback-Leibler and the Jensen-Shannon divergences, and five distribution families  denoted by $p$ and $q$ have been compared with these metrics.

\begin{figure*}[!h]
\centering
\subfigure[$p=$ Uniform distribution]{
\includegraphics[width=.45\textwidth, height=0.21\textheight]{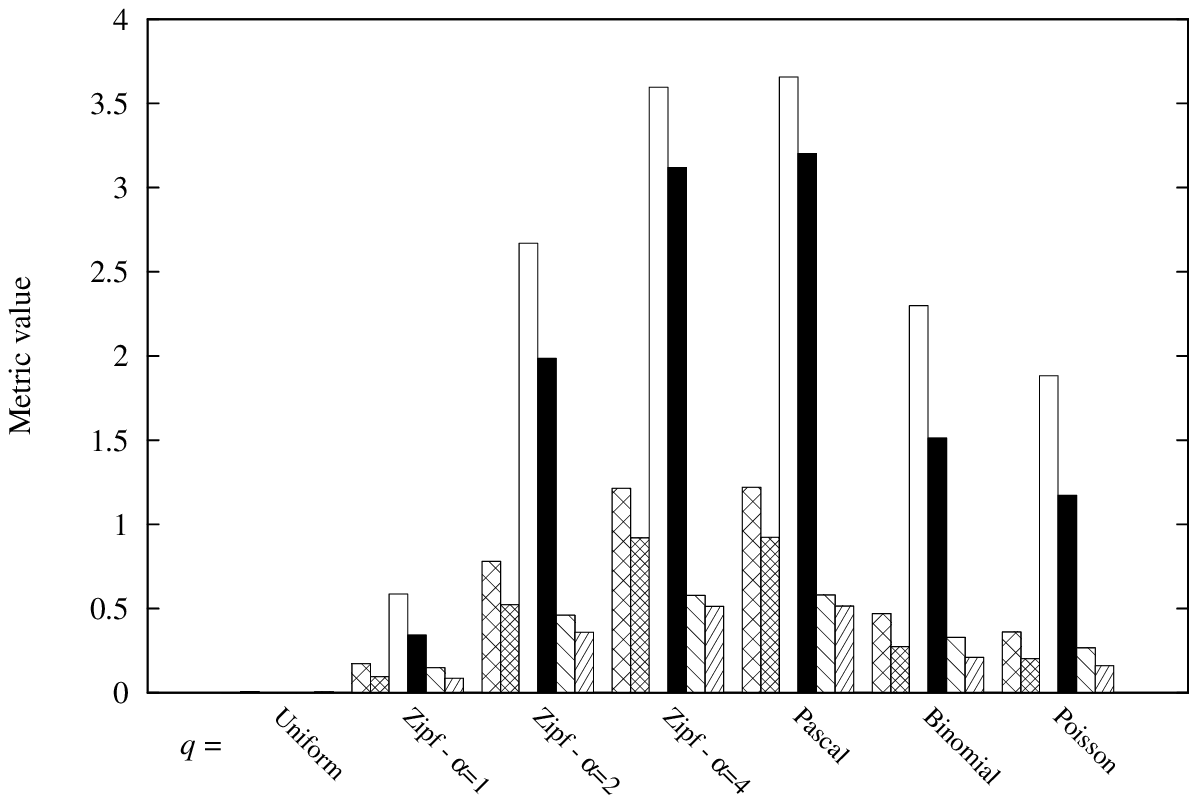}
\label{fig:test-uniform-woCCM}
}
\subfigure[$p=$ Zipf distribution with $\alpha = 1$]{
\includegraphics[width=.45\textwidth, height=0.21\textheight]{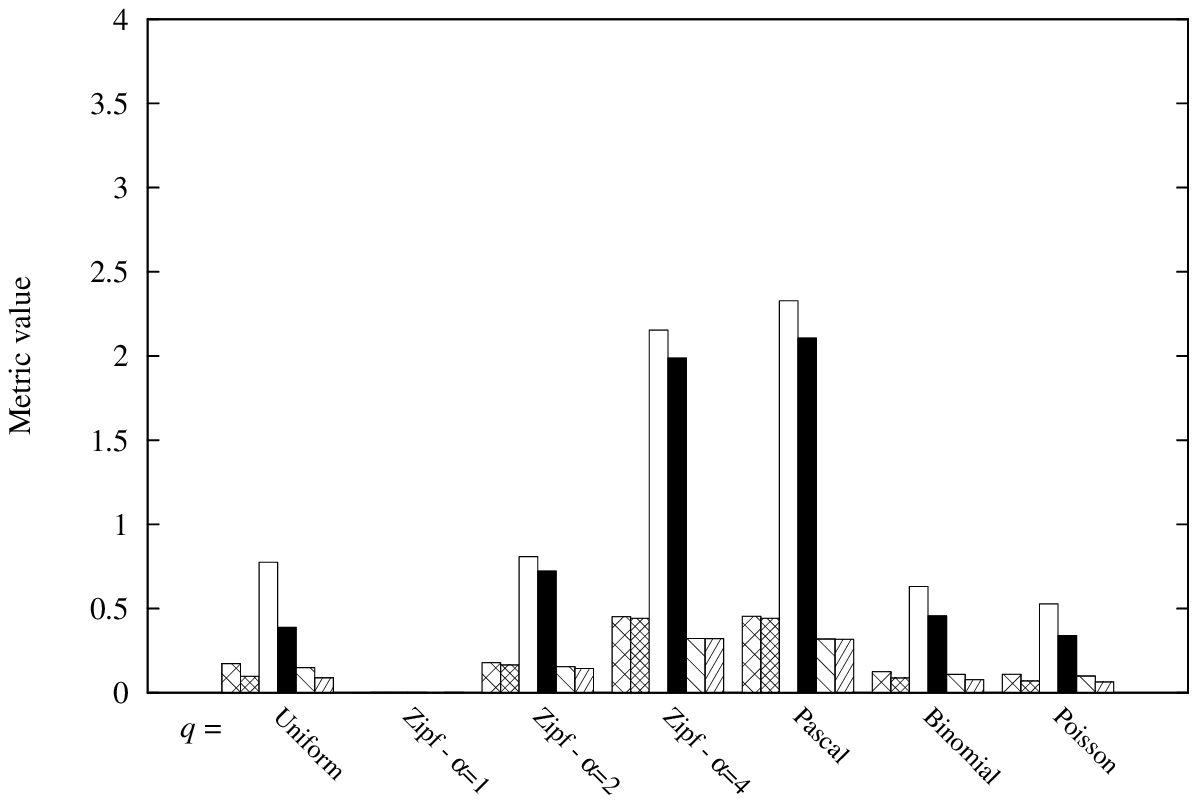}
\label{fig:test-zipf1-woCCM}
}
\subfigure[$p=$ Pascal distribution with $r=3$ and $p=\frac{n}{2r+n}$]{
\includegraphics[width=.45\textwidth, height=0.21\textheight]{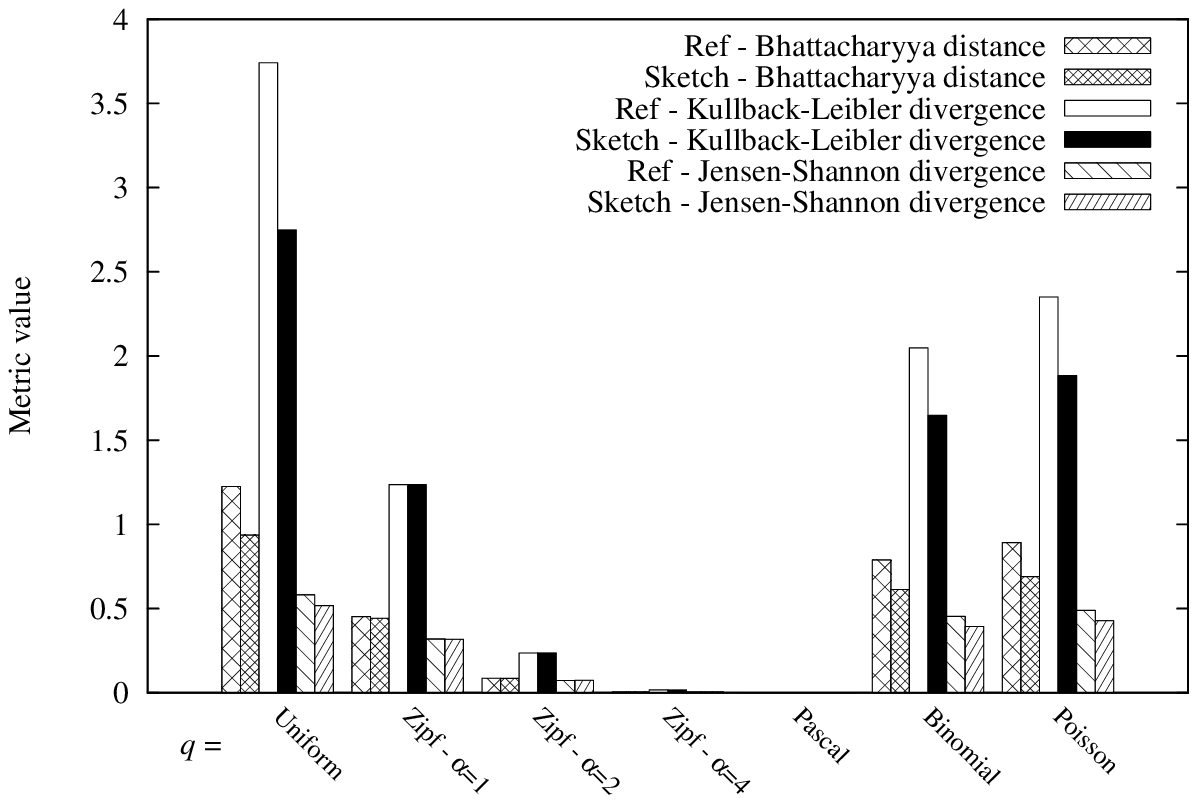}
\label{fig:test-zipf1}
}
\subfigure[$p=$ Zipf distribution with $\alpha = 2$]{
\includegraphics[width=.45\textwidth, height=0.21\textheight]{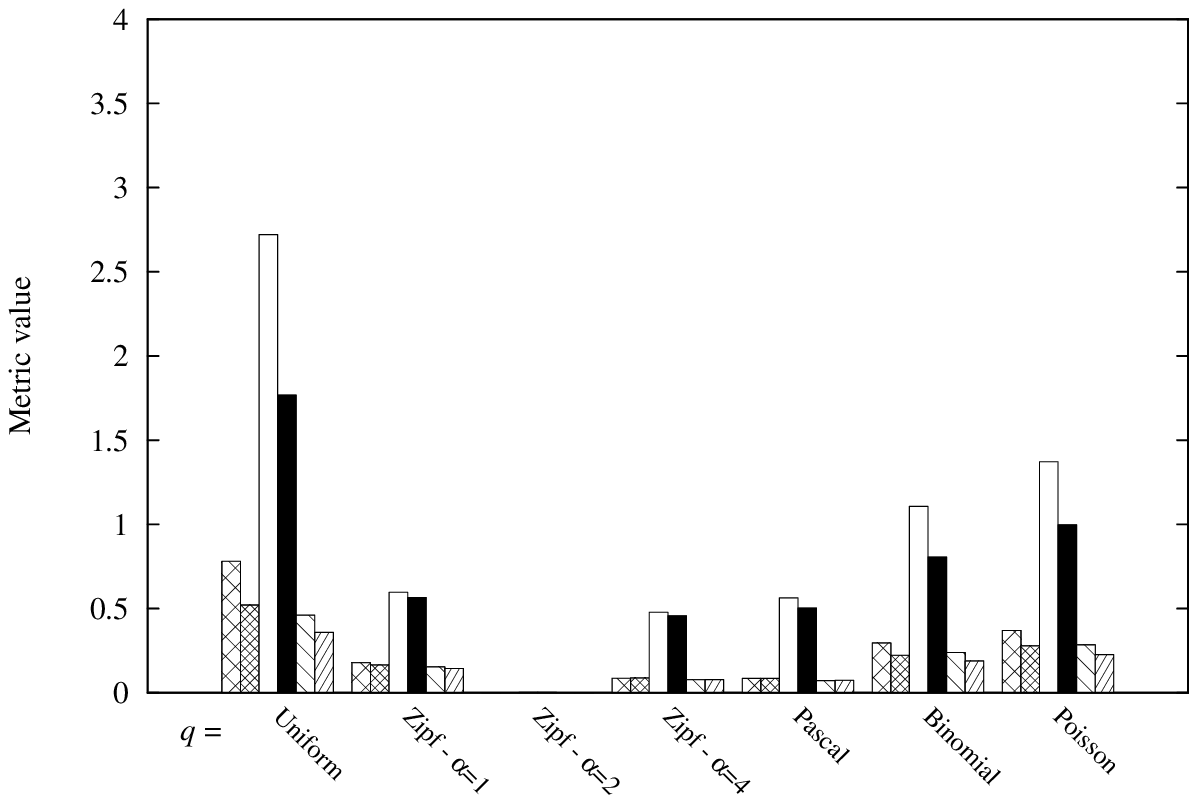}
\label{fig:test-pascal}
}
\subfigure[$p=$ Binomial distribution with $p=0.5$]{
\includegraphics[width=.45\textwidth, height=0.21\textheight]{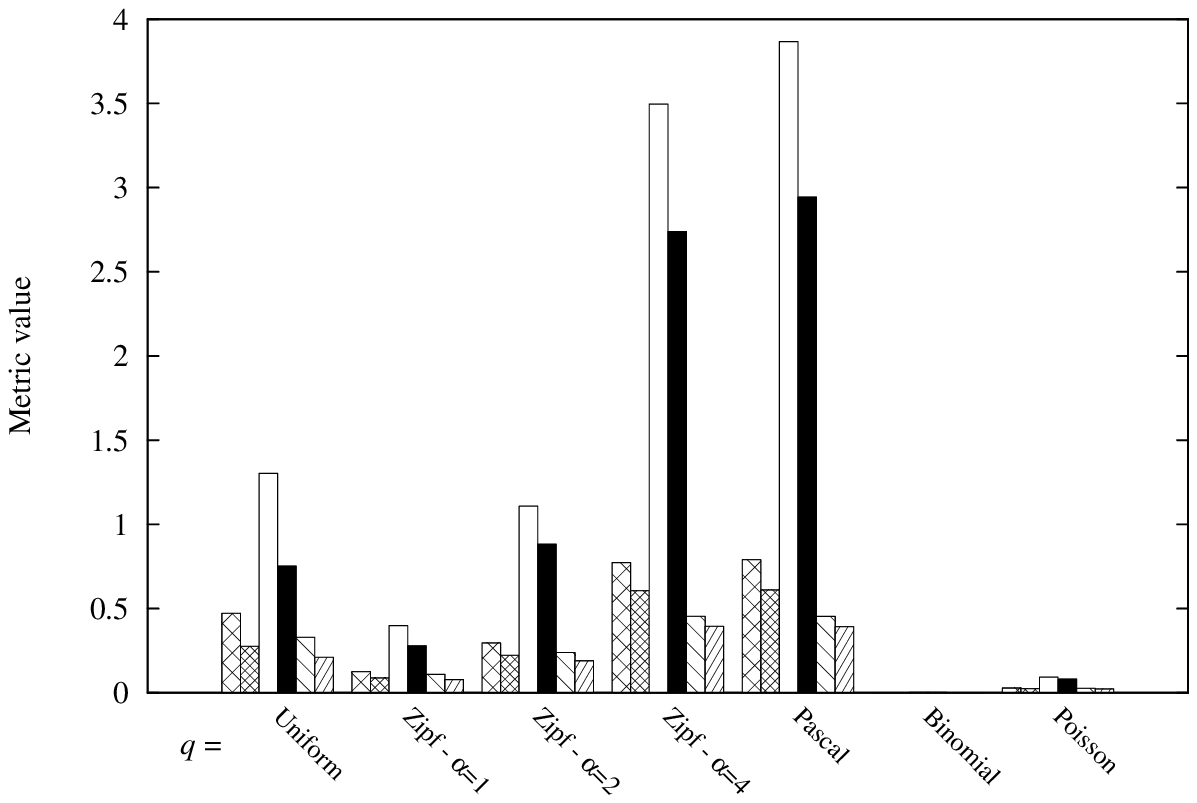}
\label{fig:test-zipf4-CCM}
}
\subfigure[$p=$ Zipf distribution with $\alpha = 4$]{
\includegraphics[width=.45\textwidth, height=0.21\textheight]{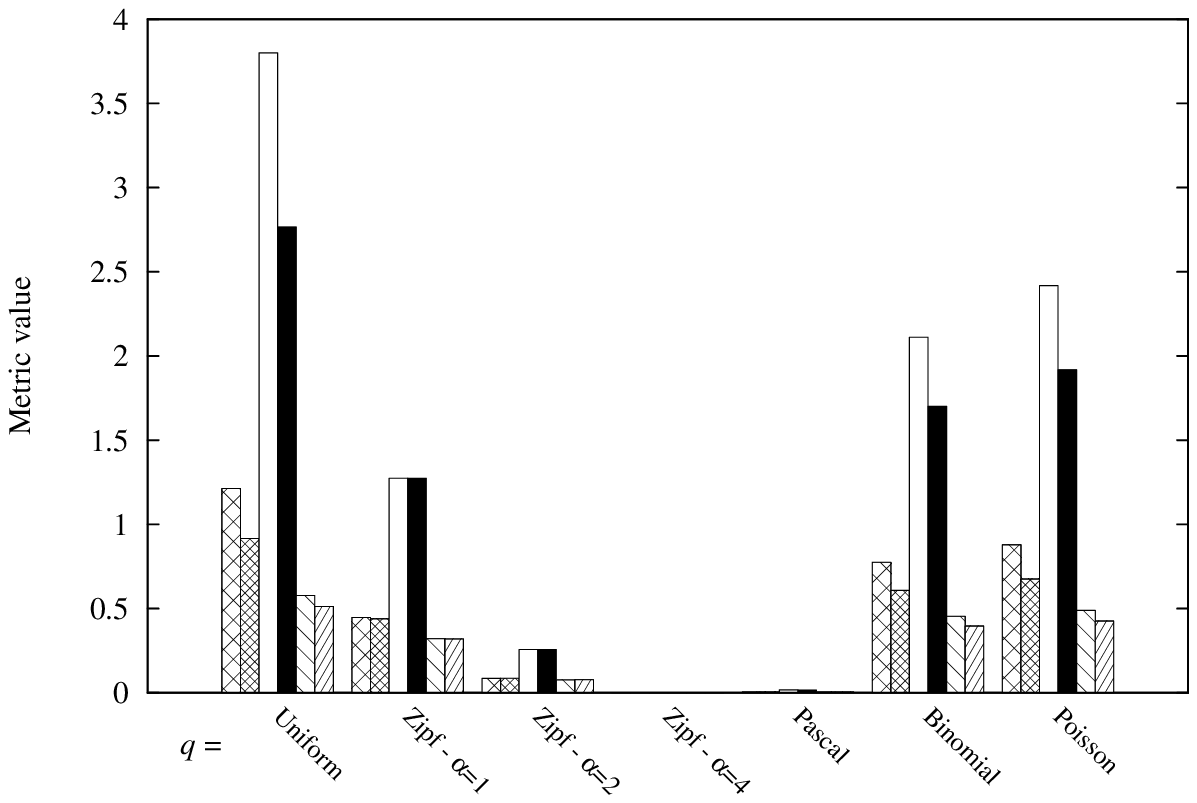}
\label{fig:test-pascal2}
}
\subfigure[$p=$ Poisson distribution with $p=\frac{n}{2}$]{
\includegraphics[width=.45\textwidth, height=0.21\textheight]{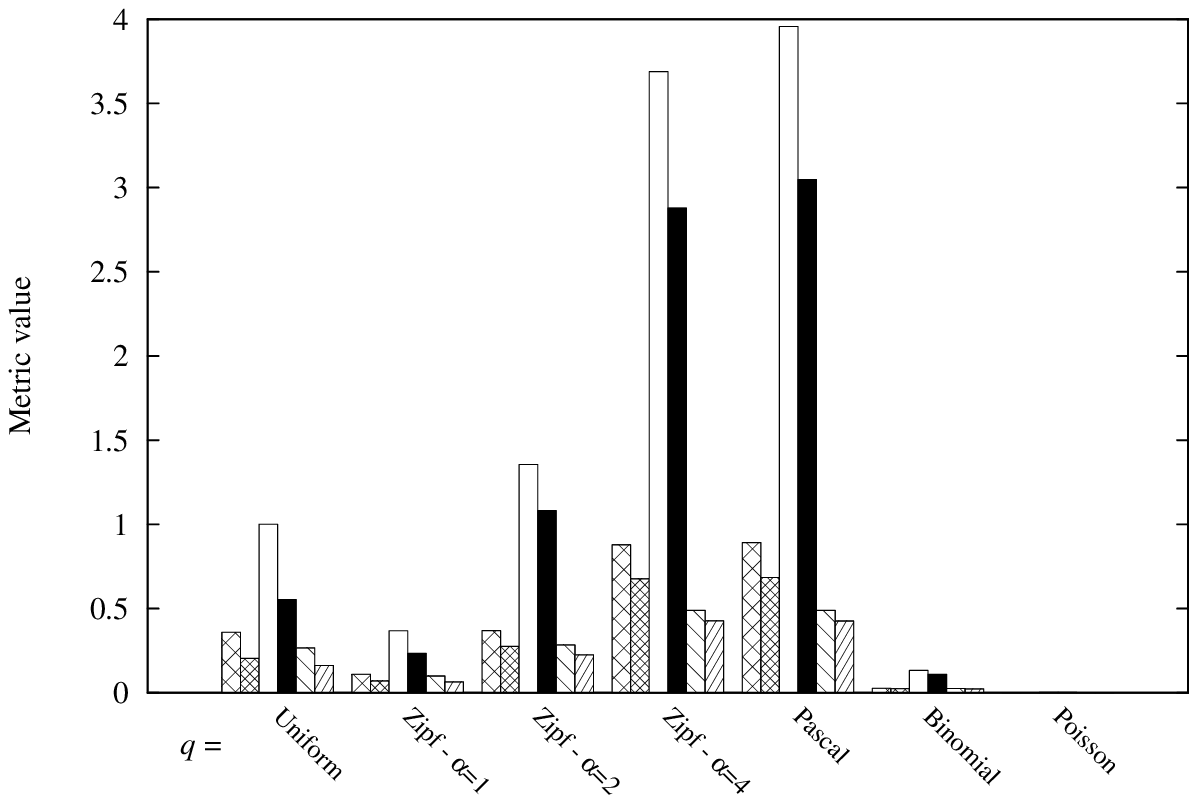}
\label{fig:test-zipf4}
}
\subfigure[$p=$~Uniform distribution and $q =$~Pascal distribution, as a function of its parameter $r$ ($p=\frac{n}{2r+n}$).]{
\includegraphics[width=.45\textwidth, height=0.21\textheight]{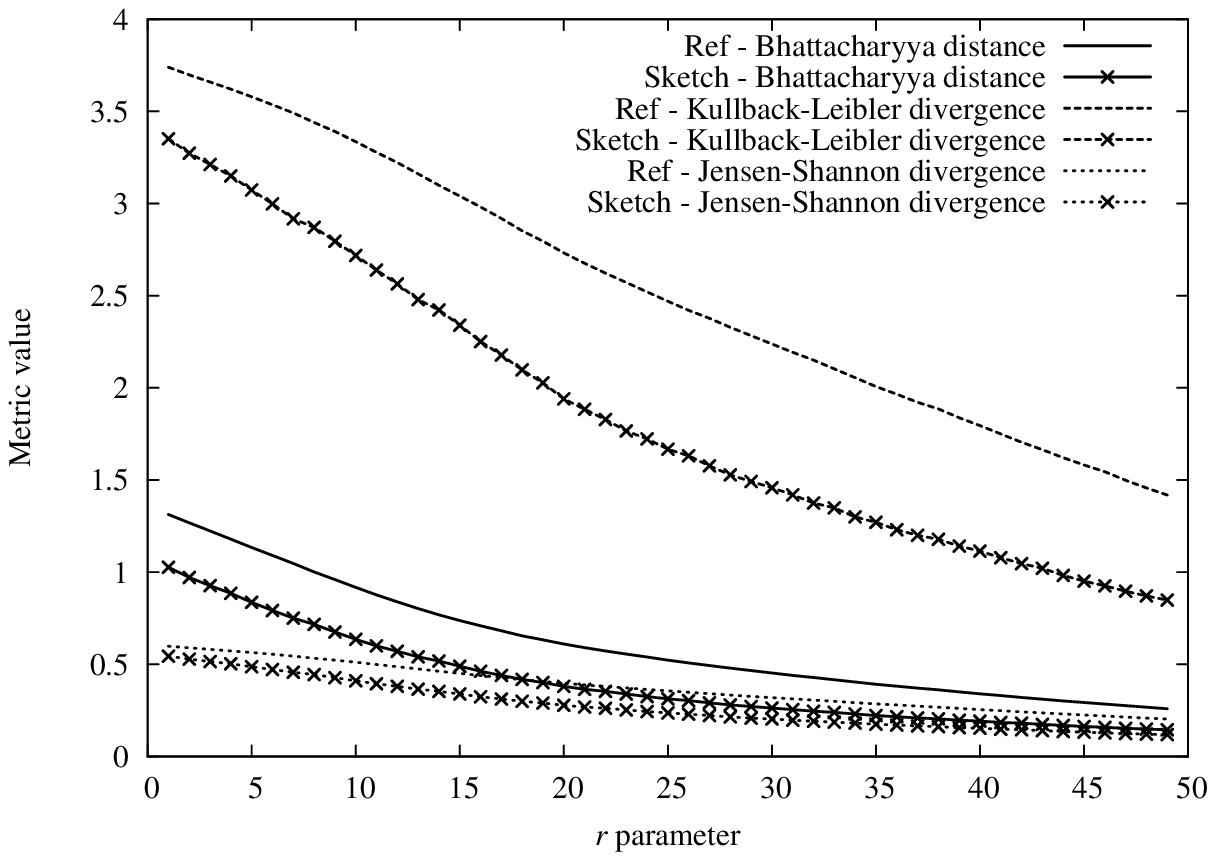}
\label{fig:test-uniform}
}
\caption{\emph{Sketch $\star$-metric} accuracy  as a function of $p$ and $q$ (or $r$ for~\ref{fig:test-uniform}). Parameters setting is as follows: $m= 200,000$; $n=4,000$; $k=200$; $t=4$ where $m$ represents the size of the stream, $n$ the number of distinct data items in the stream, $t$ the number of generated partitions and $k$ the number of cells per generated partition.}
\label{fig:accuracy}
\end{figure*}

\begin{figure*}[!h]
\centering
\subfigure[$p=$ NASA webserver (August)]{
\includegraphics[width=.45\textwidth, height=0.21\textheight]{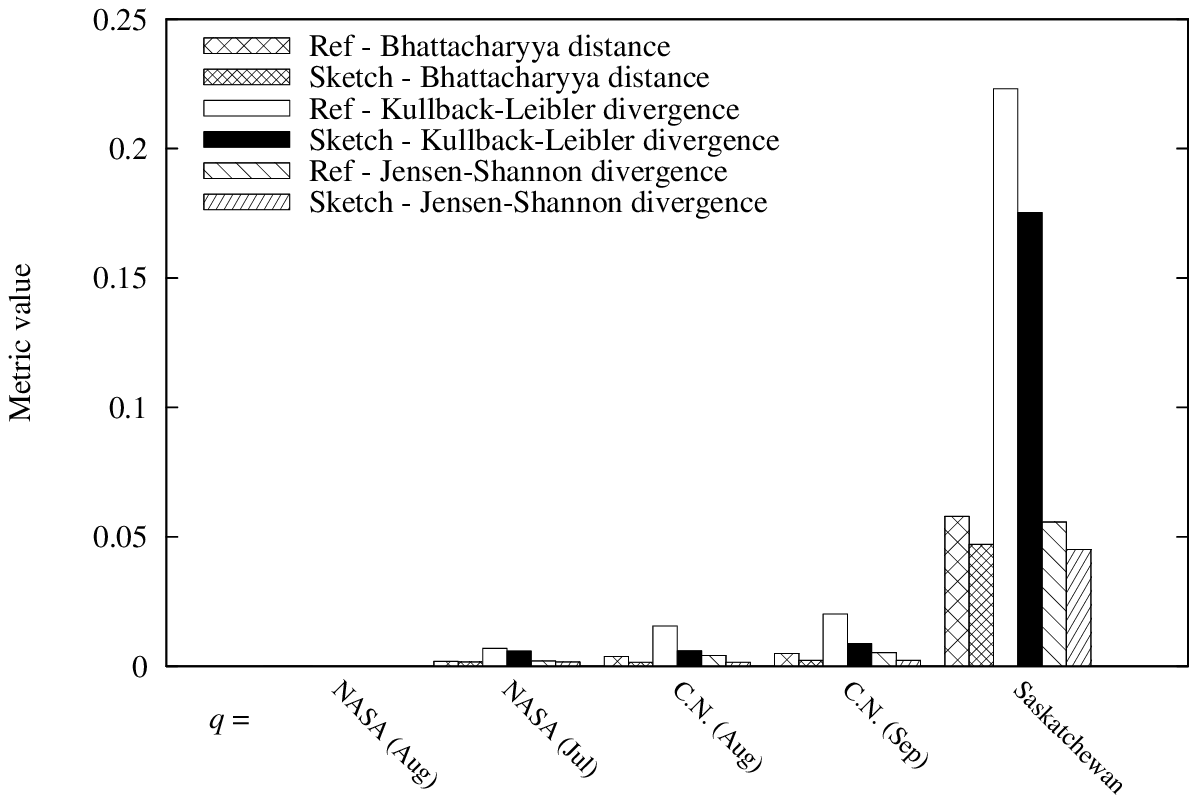}
\label{fig:test-uniform-woCCM-real}
}
\subfigure[$p=$ NASA webserver (July)]{
\includegraphics[width=.45\textwidth, height=0.21\textheight]{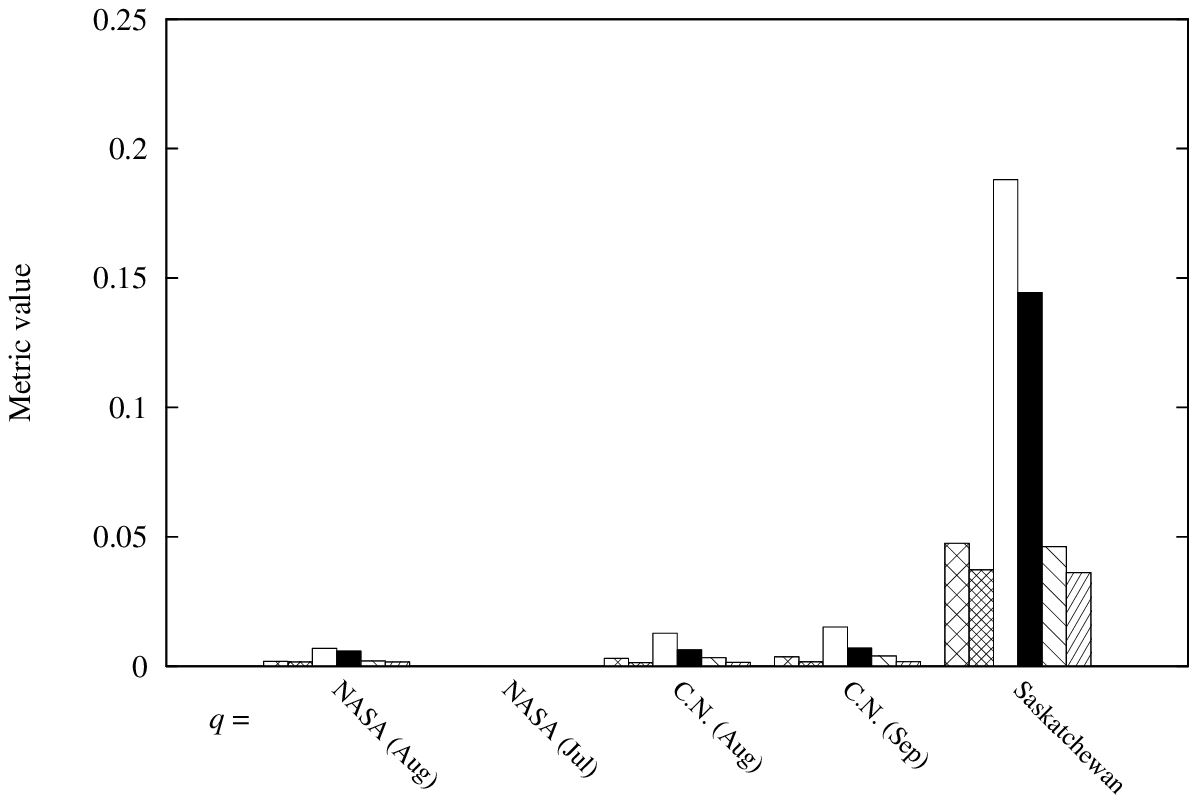}
\label{fig:test-zipf1-woCCM-real}
}
\subfigure[$p=$ ClarkNet webserver (August)]{
\includegraphics[width=.45\textwidth, height=0.21\textheight]{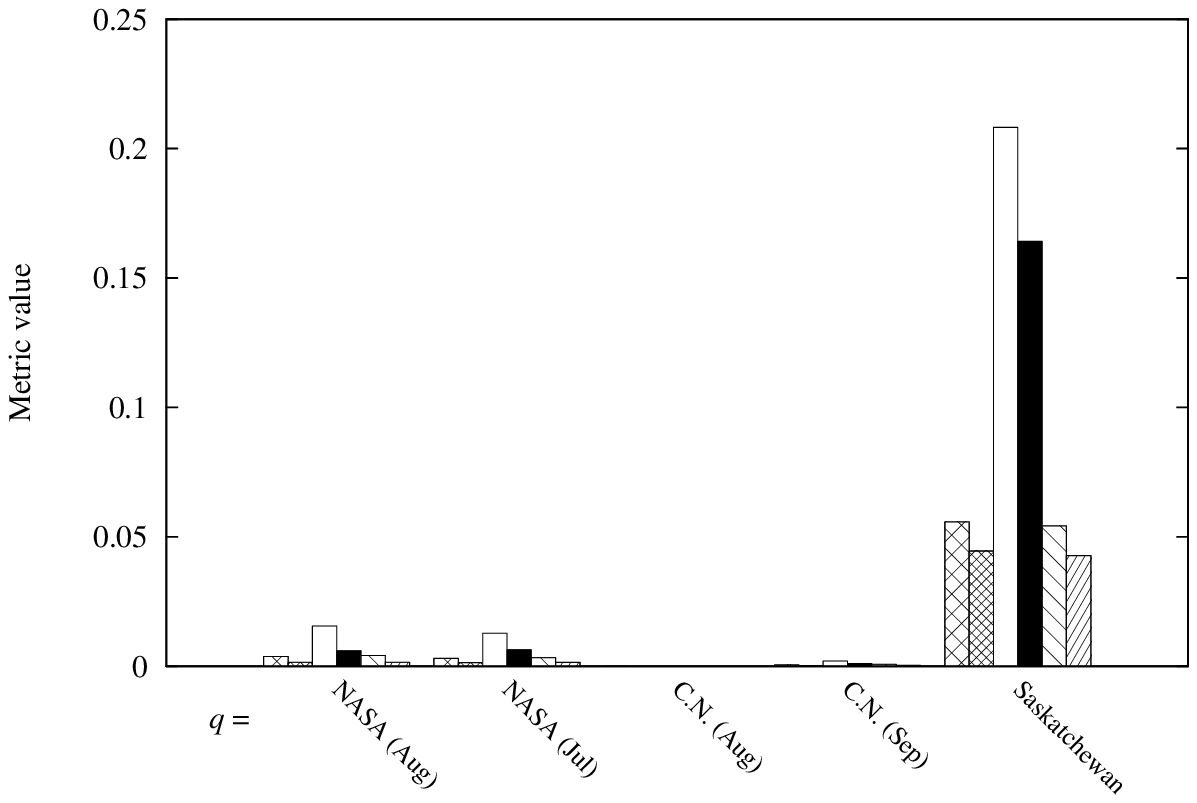}
\label{fig:test-zipf1-real}
}
\subfigure[$p=$ ClarkNet webserver (September)]{
\includegraphics[width=.45\textwidth, height=0.21\textheight]{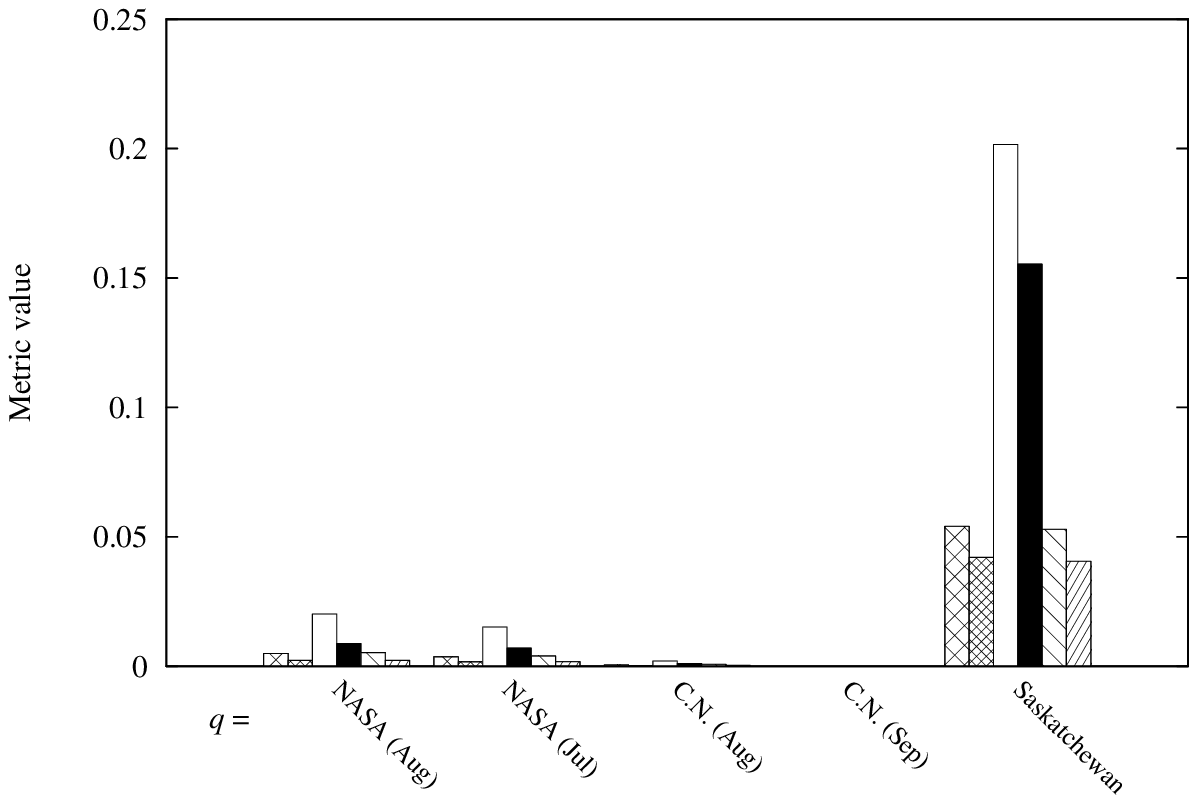}
\label{fig:test-pascal-real}
}
\subfigure[$p=$ Saskatchewan University webserver]{
\includegraphics[width=.45\textwidth, height=0.21\textheight]{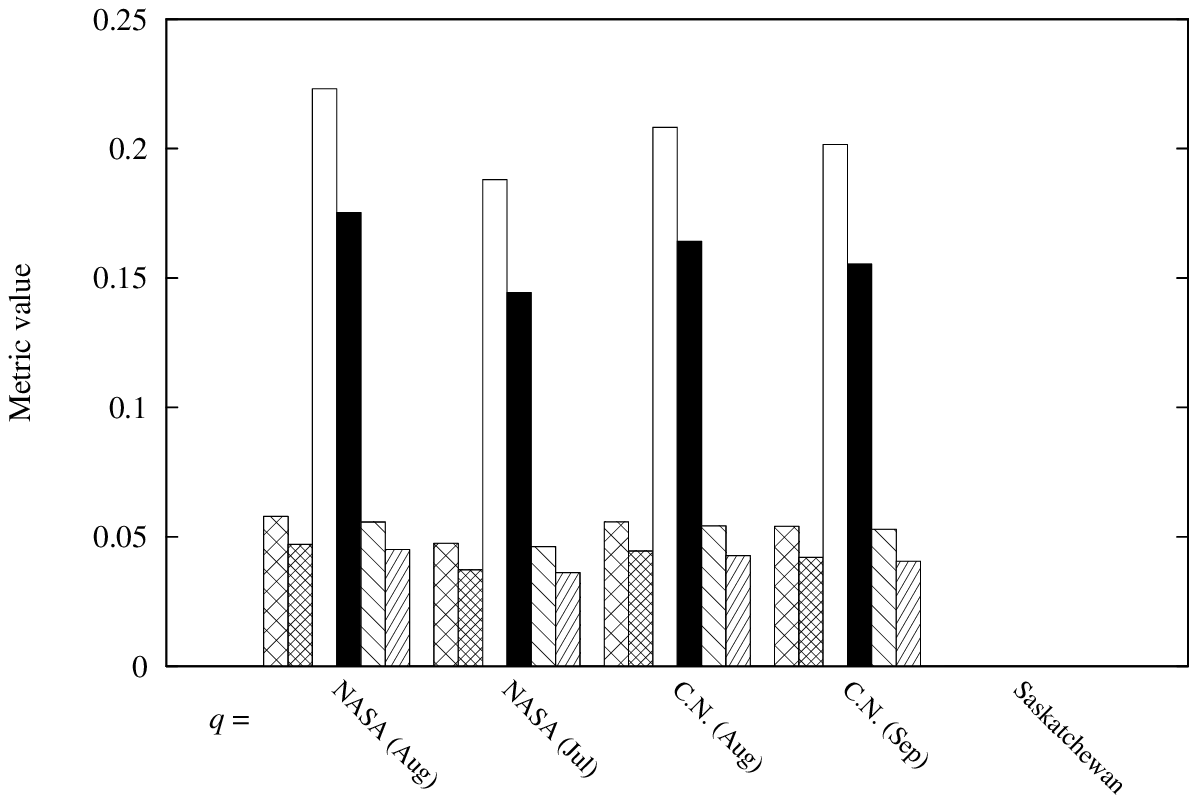}
\label{fig:test-zipf4-CCM-real}
}
\caption{\emph{Sketch $\star$-metric} accuracy  as a function of real data traces. Parameters setting is as follows: $k=2,000$; $t=4$.}
\label{fig:accuracy-real}
\end{figure*}

Let us focus on synthetic traces. The first noticeable remark is that our metric behaves perfectly well when the two compared streams follow the same distribution, whatever the generalized metric $\phi$ used (\emph{cf.}, Figure~\ref{fig:test-uniform-woCCM} with the uniform distribution, Figures~\ref{fig:test-zipf1-woCCM},~\ref{fig:test-pascal} and~\ref{fig:test-pascal2}  with the Zipf distribution, Figure~\ref{fig:test-zipf1} with the Pascal distribution, Figure~\ref{fig:test-zipf4-CCM} with the Binomial distribution, and Figure~\ref{fig:test-zipf4} with the Poisson one). This result is  interesting as it allows the sketch $\star$-metric to be a very good candidate as a  parametric method for making inference about the parameters of the distribution that follow an input stream. 
The general tendency is that when the distributions of input streams are close  (\emph{e.g}, Zipf distribution with different parameter,  Pascal and the Zipf with $\alpha=4$), then applying the generalized metrics $\phi$  on sketches give a  good estimation of the distance as computed on the full streams. 

Now, when the two input distributions exhibit a totally different shape, this may have an  impact on the precision of our metric. Specifically, let us consider as input distributions the Uniform and the Pascal distributions (see Figure~\ref{fig:test-uniform-woCCM} and~\ref{fig:test-zipf1}). Sketching the Uniform distribution leads to  $k$-cell partitions whose value is well distributed, that is, for a given partition all the $k$ cell values have with high probability the same value. Now, when sketching the Pascal distribution, the repartition of the data items in the cells of any given partitions is such that a few number of data items (those with high frequency) populate a very few number of cells. However, the values of these cells is very large compared to the other cells, which are populated by a large number of data items whose frequency is small. Thus the contribution of data items exhibiting a small frequency and sharing the cells of highly frequent items will be biased compared to the contribution of the other items. This explains why the accuracy of the sketch $\star$-metric is slightly lowered in these cases. 

We can also observe the strong impact of the non-symmetry of the Kullback-Leibler divergence on the computation of the distance (computed on full streams or on sketches) with a clear influence when the input streams follow a Pascal and Zipf with $\alpha=1$ distributions (see Figures~\ref{fig:test-zipf1} and~\ref{fig:test-zipf1-woCCM}).

Finally, Figure~\ref{fig:test-uniform} summarizes the good properties of our method whatever the  input streams to be compared and the generalized metric $\phi$ used to do this comparison. 

The same general remarks hold when considering real data sets. Indeed, Figure~\ref{fig:accuracy-real} shows that when the input streams are close to each other, which is the case for both (July and August) NASA  and (August and September) ClarkNet traces (\emph{cf.} Figure~\ref{fig:distrib-real}), then applying the generalized metrics $\phi$  on sketches gives good results w.r.t. full streams. When the shapes of the input streams are different (which is the case for Saskatchewan w.r.t. the 4 other input streams), the accuracy of the sketch $\star$-metric decreases a little bit but in a small proportion. Notice that the scales on the y-axis differ significantly in Figure~\ref{fig:accuracy} and in Figure~\ref{fig:accuracy-real}

%The impact of the number of distinct data items $n$ in the stream  on the accuracy of our metric is shown in Figure~\ref{fig:n}.  
%\begin{figure}
%\centering
%\includegraphics[width=1\linewidth]{figs/test_final_n.eps}
%\caption{\emph{Sketch $\star$-metric} accuracy  as  a function of the number of distinct items ($n$). Parameters setting is as follows: $p$ follows a Poisson distribution, and $q$ a Zipf one with $\alpha = 4$;  $m= 200.000$;  $k=200$ and $t=4$.}
%\label{fig:n}
%\end{figure}

\balance

Figure~\ref{fig:kparameters} presents the impact of the number of cells per generated partition on the accuracy of our metric on both synthetic traces and real data.  It clearly shows that, by increasing $k$,  the number of data items per cell in the generated partition shrinks and thus the absolute error  on the computation of the distance decreases. The same feature appears when the number $n$ of distinct data items in the stream increases. Indeed, when $n$ increases (for a given $k$), the number data items per cell augments and thus the precision of our metric decreases. This gives rise to a shift of the inflection point, as illustrated in Figure~\ref{fig:test_real_epsilon},  due to the fact that data sets have almost twenty times more distinct  data items than the synthetic ones. As aforementioned, the input streams exhibit very different shapes which explain the strong impact of $k$.  Note also that $k$ has the same  influence on the \emph{Sketch $\star$-metric} for all the generalized distances $\phi$. 

\FloatBarrier

\begin{figure}
\centering
\subfigure[\emph{Sketch $\star$-metric} accuracy  as  a function of  $k$. We have \newline
 $m= 200,000$; $n=4,000$; $t=4$;  $r=3$]{
\includegraphics[width=.45\textwidth, height=0.21\textheight]{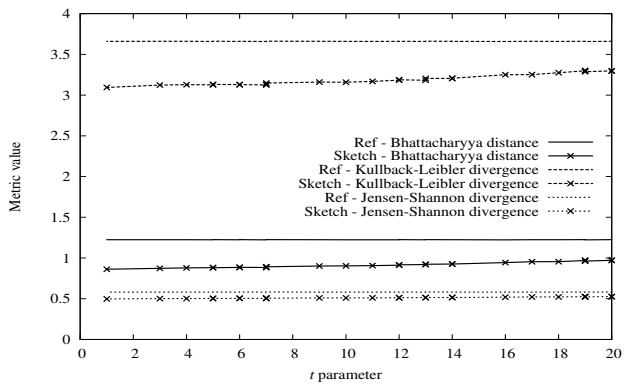}
\label{fig:test_pascal_epsilon}
}
\subfigure[{Sketch $\star$-metric} accuracy between data trace extracted from ClarkNetwork (August) and Saskatchewan University, as a function of $k$]{
\includegraphics[width=.45\textwidth, height=0.21\textheight]{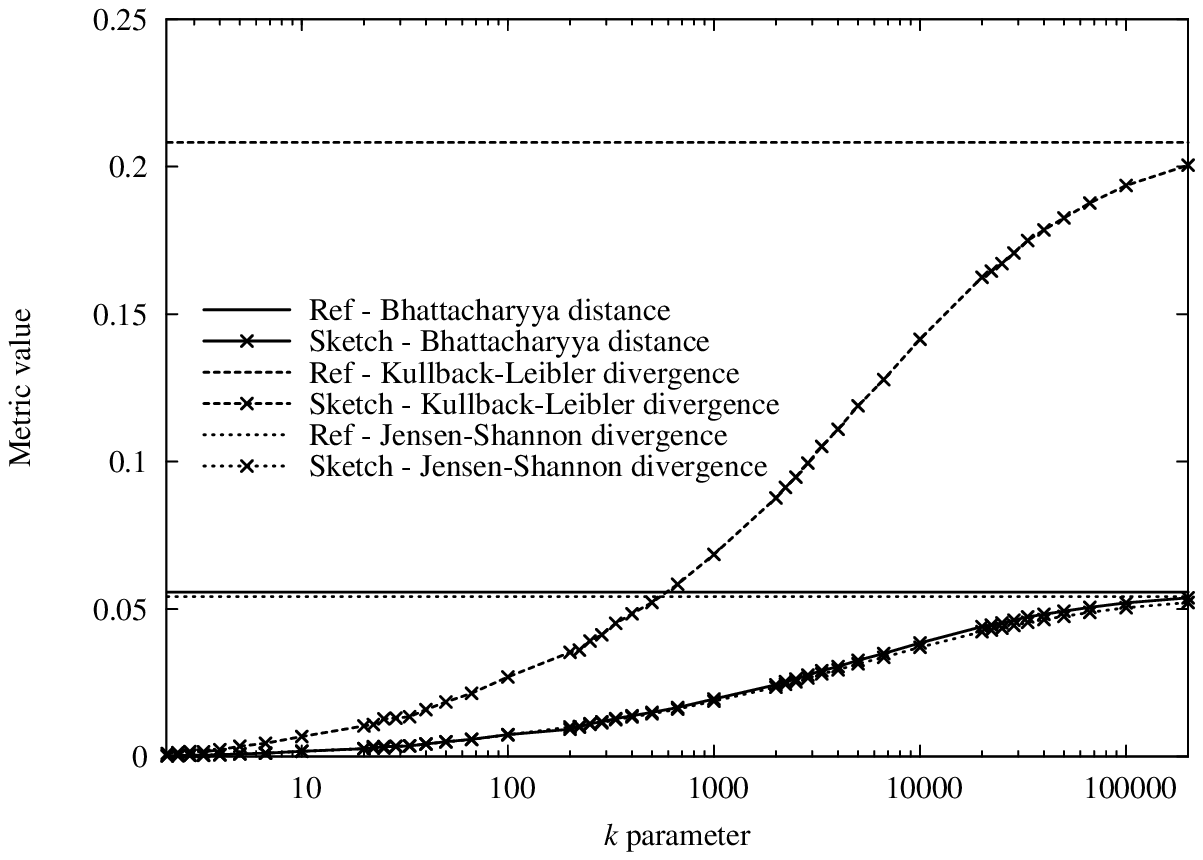}
\label{fig:test_real_epsilon}
}
\caption{\emph{Sketch $\star$-metric}  between the Uniform distribution and Pascal with parameter $p=\frac{n}{2r+n}$ (Figures~\ref{fig:test_pascal_epsilon} and~\ref{fig:test_pascal_delta}), and between data trace extracted from ClarkNetwork (August) and Saskatchewan University (Figures~\ref{fig:test_real_epsilon} and~\ref{fig:test_real_delta}).}  
\label{fig:kparameters}
\end{figure}

Figure~\ref{fig:tparameters} shows the slight  influence of the number $t$ of generated partitions on the accuracy of our metric. The reason comes from the use of $2$-universal hash functions, which  guarantee  for each of them and with high probability that data items are uniformly distributed over the cells of any partition. As a consequence, augmenting the number of such hash functions has a weak influence on the accuracy of the metric.  Figure~\ref{fig:parameters} focuses on  the error made on the \emph{Sketch $\star$-metric} for five different values of $t$  as a function of parameter $r$ of the Pascal distribution (recall that increasing values of $r$ -- while maintaining the mean value -- makes the shape of the Pascal  distribution flatter).  Figures~\ref{fig:test_pascal_r_delta_b_error},~\ref{fig:test_pascal_r_delta_kl_error}, and~\ref{fig:test_pascal_r_delta_js_error} respectively depict for each value of $t$ the difference between the reference and the sketch values which makes more visible the impact of $t$. The same main lesson drawn from these figures is the moderate  impact of $t$ on the precision of our algorithm.

%\begin{figure*}[!h]
%\centering
%\subfigure{
%\includegraphics[width=.45\textwidth, height=0.21\textheight]{figs/test_real_epsilon.eps}
%\label{fig:test_real_epsilon}
%}
%\subfigure{
%\includegraphics[width=.45\textwidth, height=0.21\textheight]{figs/test_real_delta.eps}
%\label{fig:test_real_delta}
%}
%\caption{\emph{Sketch $\star$-metric} accuracy between data trace extracted from ClarkNetwork (August) and Saskatchewan University, as a function of $k$ and $t$.}  
%\label{fig:parameters-real}
%\end{figure*}

%% file: conclusion.tex
%!TEX root =  AB13-INFOCOM-RR.tex

\section{Conclusion and open issues}\label{sec:conclusion}

In this paper, we have introduced a new metric, the \emph{Sketch $\star$-metric}, that allows to compute any  generalized metric $\phi$ on  the summaries of two large input streams. We have presented  a simple and efficient algorithm  to sketch streams and compute this metric, and we have shown that it behaves pretty well whatever the considered input streams. We are convinced of the undisputable  interest of such a metric in various domains including machine learning, data mining, databases, information retrieval and network monitoring. 

Regarding future  works, we plan to characterize  our metric among  R\'enyi divergences~\cite{renyi61}, also known as $\alpha$-diver\-gences, which generalize different divergence classes. We also plan to consider a distributed setting, where each site would be in charge of analyzing its own streams and then would propagate its results to the other sites of the system for comparison or merging. An immediate application of such a ``tool'' would be to detect massive attacks in a decentralized manner (\emph{e.g.}, by identifying specific connection profiles as with worms propagation, and massive port scan attacks or by detecting sudden variations in the volume of received data).